\documentclass[11pt,letterpaper]{article}
\usepackage[margin=1in]{geometry}

\usepackage{latexsym,amssymb}
\usepackage{amsmath,enumerate}

\usepackage[numbers]{natbib}

\usepackage{subfigure}
\usepackage{float}
\usepackage{wrapfig}
\usepackage{psfig}
\usepackage{epsfig}
\usepackage{xspace}
\usepackage{paralist}
\usepackage{enumerate}
\usepackage{cases}
\usepackage{caption}
\usepackage{mathtools}

\usepackage{algorithm}
\usepackage{algorithmicx}
\usepackage{algpseudocode}
\usepackage{authblk}
\usepackage{tikz}

\usepackage{color}
\definecolor{Darkblue}{rgb}{0,0,0.4}
\definecolor{Brown}{cmyk}{0,0.81,1.,0.60}
\definecolor{Purple}{cmyk}{0.45,0.86,0,0}
\usepackage[breaklinks]{hyperref}
\hypersetup{colorlinks=true,%pdfborder={1 1 1 [3]},%
            citebordercolor={.6 .6 .6},linkbordercolor={.6 .6 .6},%
citecolor=black,urlcolor=black,linkcolor=black,pagecolor=black}
\newcommand{\lref}[2][]{\hyperref[#2]{#1~\ref*{#2}}}

\usepackage{multicol}

\newenvironment{proof}{{\bf Proof:  }}{\hfill\rule{2mm}{2mm}}
\newenvironment{proofof}[1]{{\bf Proof of #1:  }}{\hfill\rule{2mm}{2mm}}

\newtheorem{definition}{Definition}[section]

\newtheorem{theorem}{Theorem}[section]
\newtheorem{lemma}{Lemma}[section]

\newtheorem{theorem*}{Theorem}
\newtheorem{lemma*}{Lemma}

\newtheorem{corollary*}{Corollary}
\newtheorem{definition*}{Definition}

\newcommand{\ignore}[1]{}

\newcommand{\na}{n(A)}
\newcommand{\ppointa}[2]{{#1} = ( {#2}_1,{#2}_2,\ldots,{#2}_{\na})} % Order

\newcommand{\cpx}[1]{\mathcal{O}(#1)} % Order
\newcommand{\blotto}{\mathcal{B}}

\newcommand{\Rna}{\mathbb{R}^{\na}}

\newcommand{\pr}[2]{\text{Pr}(#1 = #2)}
\newcommand{\functiona}{\mathcal{G}_A}
\newcommand{\functionb}{\mathcal{G}_B}
\newcommand{\strategyset}{\mathcal{M}(\mathcal{X})}
\newcommand{\pureset}{\mathcal{X}}
\newcommand{\strategysetb}{\mathcal{M}(\mathcal{Y})}
\newcommand{\puresetb}{\mathcal{Y}}
\newcommand{\sets}{S_A}
\newcommand{\spt}{\mathcal{H}(S_A)}
\newcommand{\vs}{I_A}

\newcommand{\vv}{\hat{\mathcal{X}}}

\newcommand{\poly}{\operatorname{poly}}

\newcommand{\sign}{\operatorname{sign}}
\newcommand{\minpure}{\mbox{\textsc {FindBestPure}} }
\newcommand{\hporacle}{\mbox{\textsc {HyperplaneOracle}} }
\newcommand{\broracle}{\mbox{\textsc {BestRespOracle}} }

\newcommand{\x}{\mathbf {x}}
\newcommand{\xp}{\mathbf {x'}}
\newcommand{\xz}{\mathbf {x''}}
\newcommand{\xhp}{\mathbf{\hat{x}'}}
\newcommand{\xhz}{\mathbf{\hat{x}''}}

\newcommand{\xh}{\mathbf{\hat{x}}}
\newcommand{\y}{\mathbf{y}}
\newcommand{\yh}{\mathbf{\hat{y}}}
\newcommand{\B}{\mathcal{B}}
\newcommand{\G}{\functiona}
\newcommand{\X}{\pureset}
\newcommand{\Y}{\mathcal{Y}}
\newcommand{\Gb}{\functionb}
\newcommand{\n}{n(A)}
\newcommand{\nb}{n(B)}
\newcommand{\Sa}{\sets}
\newcommand{\Sb}{S_B}
\newcommand{\Ia}{\vs}
\newcommand{\Ib}{I_B}
\newcommand{\ut}{u^T}
\newcommand{\um}{u^M}
\newcommand{\E}{\mathbb{E}}

\newcommand{\hx}{\hat{x}}
\newcommand{\argmin}{\operatornamewithlimits{argmin}}

\newcommand{\px}{x}
\newcommand{\py}{y}
\newcommand{\MX}{\strategyset}
\newcommand{\MY}{\strategysetb}
\newcommand{\distance}{bounded distance}

{\makeatletter
 \gdef\xxxmark{%
   \expandafter\ifx\csname @mpargs\endcsname\relax % in minipage?
     \expandafter\ifx\csname @captype\endcsname\relax % in figure/caption?
       \marginpar{xxx}% not in a caption or minipage, can use marginpar
     \else
       xxx % notice trailing space
     \fi
   \else
     xxx % notice trailing space
   \fi}
 \gdef\xxx{\@ifnextchar[\xxx@lab\xxx@nolab}
 \long\gdef\xxx@lab[#1]#2{{\bf [\xxxmark #2 ---{\sc #1}]}}
 \long\gdef\xxx@nolab#1{{\bf [\xxxmark #1]}}
}

\title{
From Duels to Battlefields: Computing Equilibria of Blotto and Other Games\footnote{Supported in part by NSF CAREER award 1053605, NSF grant CCF-1161626, ONR YIP award N000141110662, DARPA GRAPHS/AFOSR grant FA9550-12-1-0423, and a Google faculty research award.}
}

\begin{document}

\author[1]{AmirMahdi Ahmadinejad}
\author[2]{Sina Dehghani}
\author[2]{MohammadTaghi Hajiaghayi}
\author[3]{Brendan Lucier}
\author[2]{Hamid Mahini}
\author[2]{Saeed Seddighin}
\affil[1]{Department of Management Science and Engineering, Stanford University\newline  \url{amirmahdi.ahmadi@gmail.com}}
%\affil[2]{Department of Computer Science, Univeristy of Maryland}
%\affil[3]{Microsoft Research}
%\affil[1]{Department of Computer Engineering, Sharif University of Technology\newline  \url{amirmahdi.ahmadi@gmail.com}}
\affil[2]{Department of Computer Science, Univeristy of Maryland\newline
 \url{dehghani, hajiagha, hmahini, sseddigh@cs.umd.edu}}
\affil[3]{Microsoft Research \hspace{1cm}\url{brlucier@microsoft.com}}

\date{}

\begin{titlepage}
\maketitle
\vspace{-5mm}
\begin{abstract}
We study the problem of computing Nash equilibria of zero-sum games.
Many natural zero-sum games have exponentially many strategies, but highly structured payoffs.  For example, in the well-studied Colonel Blotto game (introduced by Borel in 1921), players must divide a pool of troops among a set of battlefields with the goal of winning (i.e., having more troops in) a majority.  
The Colonel Blotto game is commonly used for analyzing a wide range of applications from the U.S presidential election, to innovative technology competitions, to
advertisement, to sports.
However, because of the size of the strategy space, standard  
methods for computing equilibria of zero-sum games fail to be computationally feasible.
Indeed, despite its importance, only a few solutions for special variants of the problem are known. 

In this paper we show how to compute equilibria of Colonel Blotto games.

Moreover, our approach takes the form of a general reduction: to find a Nash equilibrium of a zero-sum game, it suffices to design a separation oracle for the strategy polytope of any bilinear game that is payoff-equivalent.  
We then apply this technique to obtain the first polytime algorithms for a variety of games.  In addition to Colonel Blotto, we also show how to compute equilibria in an infinite-strategy variant called the General Lotto game; this involves showing how to prune the strategy space to a finite subset before applying our reduction.  We also consider the class of dueling games, first introduced by Immorlica et al. (2011).  We show that our approach provably extends the class of dueling games for which equilibria can be computed: we introduce a new dueling game, the matching duel, on which prior methods fail to be computationally feasible but upon which our reduction can be applied.

\end{abstract}
\thispagestyle{empty}
\end{titlepage}

%\maketitle
%
%\begin{abstract}
%\input{abstract.tex}
%\end{abstract}

%\author[1]{AmirMahdi Ahmadinejad}
%\author[1]{AmirMahdi Ahmadi}
%\author[2]{Sina Dehghani}
%\author[2]{MohammadTaghi Hajiaghayi}
%\author[3]{Brendan Lucier}
%\author[2]{Hamid Mahini}
%\author[2]{Saeed Seddighin}
%\affil[1]{Department of Computer Engineering, Sharif University of Technology} %\\ \texttt{amirmahdi.ahmadi@gmail.com}}
%\affil[2]{Department of Computer Science, University of Maryland} %\\ \texttt{dehghani, hajiagha, hmahini, sseddigh@cs.umd.edu}}
%\affil[3]{Microsoft Research} %\\ \texttt{brlucier@microsoft.com}}

%\affil[1]{Department of Computer Engineering, Sharif University of Technology \newline  \url{amirmahdi.ahmadi@gmail.com}}
%\affil[2]{Department of Computer Science, Univeristy of Maryland\newline
% \url{dehghani, hajiagha, hmahini, sseddigh@cs.umd.edu}}
%\affil[3]{Microsoft Research\newline \url{brlucier@microsoft.com}}

%\begin{bottomstuff}
%This work is supported in part by NSF CAREER award 1053605, NSF grant CCF-1161626, ONR YIP award N000141110662, and a DARPA/AFOSR grant FA9550-12-1-0423.
%\end{bottomstuff}

%\maketitle

%\newpage

\section{Introduction}

%BRENDAN: Old paragraph... 
%The problem of finding Nash equilibria is one of the fundamental problems in Algorithmic Game Theory, and has made a nice connection between Game Theory and Computer Science \cite{papa94}. While Nash~\cite{nash} introduced the concept of Nash equilibrium and proved every game has a mixed Nash equilibrium, finding a Nash equilibrium in reasonable time seems to be essential. The Nash equilibrium can be used as a basic equilibrium concept, only if it is efficiently computable such that it can be used for predicting the outcome of a real-world game. This highlights the role of computer scientists in this area which aim to design efficient algorithms for finding Nash equilibria of various games (see, e.g., \cite{vetta02,Chen06,papa07,aaems08,rout09,das09,papa09}). In this paper, we follow this line of research and propose polynomial-time algorithms for finding Nash equilibria of the Colonel Blotto and the General Lotto games. Moreover, we show our techniques can be used for finding a Nash equilibrium of a broad range of zero-sum games.

%BRENDAN: New paragraph... 
%% Finding Nash Equilibria
Computing a Nash equilibrium of a given game is a central problem in algorithmic game theory.
It is known that every finite game admits a Nash equilibrium (that is, a profile of strategies from which no player can benefit from a unilateral deviation) \cite{Nash51}.  But 
%Moreover, it is common practice to take such equilibria as predictive of play.  
it is not necessarily obvious how to find an equilibrium.  Indeed, the conclusions to date have been largely negative:
%Given a game, how should one find a Nash equilibrium?  This fundamental question in algorithmic game theory 
%The conclusions to date have been largely negative: 
computing a Nash equilibrium of a normal-form game is known to be PPAD-complete \cite{DGP06,GP06}, even for two-player games \cite{CD06}.  
In fact, it is PPAD-complete to find an $\frac{1}{n^{O(1)}}$ approximation to a Nash equilibrium \cite{CDT06}.  These results call into question the predictiveness of Nash equilibrium as a solution concept.  

This motivates the study of classes of games for which equilibria can be computed efficiently.  It has been found that many natural and important classes of games have structure that can be exploited to admit computational results \cite{Dantzig63,GJM11,KS10,LMM03}.  Perhaps the most well-known example is the class of zero-sum two-player games\footnote{Or, equivalently, constant-sum games.}, where player $2$'s payoff is the negation of player $1$'s payoff.  The normal-form representation of a zero-sum game is a matrix $A$, which specifies the game payoffs for player 1. 
%(with player 2's payoff being the negation).  
This is a very natural class of games, as it models perfect competition between two parties.  Given the payoff matrix for a zero-sum game as input, a Nash equilibrium can be computed in polynomial time, and hence time polynomial in the number of pure strategies available to each player \cite{Dantzig63}.  Yet even for zero-sum games, this algorithmic result is often unsatisfactory.  The issue is that for many games the most natural representation is more succinct than simply listing a payoff matrix, so that the number of strategies is actually exponential in the most natural input size.  In this case the algorithm described above fails to guarantee efficient computation of equilibria, and alternative approaches are required.

%%%
% An example: Blotto
%%%

\paragraph{The Colonel Blotto Game}
A classical and important example illustrating these issues is the Colonel Blotto game, first introduced by Borel in 1921 
%\cite{B21} and popularized by a dedicated issue of Econometrica in 1943 \cite{B53,F53a,F53b,V53}.
\cite{B21,B53,F53a,F53b,V53}.
In the Colonel Blotto game, two colonels each have a pool of troops and must fight against each other over a set of battlefields. The colonels simultaneously divide their troops between the battlefields.
A colonel wins a battlefield if the number of his troops dominates the number of troops of his opponent. 
The final payoff of each colonel is the (weighted) number of battlefields won.
An equilibrium of the game is a pair of colonels' strategies, which is a (potentially randomized) distribution of troops across battlefields, such that no colonel has incentive to change his strategy.
Although the Colonel Blotto game was initially proposed to study a war situation, it has found applications in the analysis of many different forms of competition: from sports, to advertisement, to politics \cite{M93,LP02,MMT05,CKS09,KR10,KR12}, and has thus become one of the most well-known games in classic game theory.

%The problem of finding Nash equilibria is one of the fundamental problems in Algorithmic Game Theory, and has made a nice connection between Game Theory and Computer Science \cite{papa94}. While Nash~\citeyear{nash} introduced the concept of Nash equilibrium and proved every game has a mixed Nash equilibrium, finding a Nash equilibrium in reasonable time seems to be essential. The Nash equilibrium can be used as a basic equilibrium concept, only if it is efficiently computable such that it can be used for predicting the outcome of a real-world game. This highlights the role of computer scientists in this area which aim to design efficient algorithms for finding Nash equilibria of various games (see, e.g., \cite{vetta02,Chen06,papa07,aaems08,rout09,das09,papa09}). In this paper, we follow this line of research and propose polynomial-time algorithms for finding Nash equilibriums of the Colonel Blotto and the General Lotto games.

Colonel Blotto is a zero-sum game.  However, the number of strategies in the Colonel Blotto game is exponential in its natural representation.  After all, there are ${n+k-1 \choose k-1}$ ways to partition $n$ troops among $k$ battlefields.  The classical methods for computing the equilibra of a zero-sum game therefore do not yield computationally efficient results.  Moreover, significant effort has been made in the economics literature to understand the structure of equilibria of the Colonel Blotto game, i.e., by solving for equilibrium explicitly \cite{T49,Bl54,Bl58,Be69,SW81,W05,R06,K07,H07,GP09,KR12}.  Despite this effort, progress remains sparse.  Much of the existing work considers a continuous relaxation of the problem where troops are divisible, and for this relaxation a significant breakthrough came only quite recently in the seminal work of \citeauthor{R06}~\cite{R06}, 85 years after the introduction of the game.  \citeauthor{R06} finds an equilibrium solution for the continuous version of the game, in the special case that all battlefields have the same weight.  The more general weighted version of the problem remains open, as does the original non-relaxed version with discrete strategies. Given the apparent difficulty of solving for equilibrium explicitly, it is natural to revisit the equilibrium computation problem for Colonel Blotto games.  

%We show how to apply our general reduction to this class of games, resulting in the first (and, to the best of our knowledge, simplest) polynomial-time algorithm for computing equilibria of the Colonel Blotto game.  

\paragraph{An Approach: Bilinear Games}
How should one approach equilibrium computation in such a game?
The exponential size of the strategy set is not an impassable barrier; in certain cases, games with exponentially many strategies have an underlying structure that can be used to approach the equilibrium computation problem.  For example, Koller, Megiddo and von Stengel \cite{KMvS94} show how to compute equilibria for zero-sum extensive-form games with perfect recall.  Immorlica et al. \cite{duel11} give an approach for solving algorithmically-motivated ``dueling games'' with uncertainty.  Letchford and Conitzer \cite{LC13} compute equilibria for a variety of graphical security games. Each of these cases involve games with exponentially many strategies.  In each case, a similar approach is employed: reformulating the original game as a payoff-equivalent \emph{bilinear} game.  In a bilinear game, the space of strategies forms a polytope in $R^n$, and payoffs are specified by a matrix $M$: if the players play strategies $x$ and $y$ respectively, then the payoff to player $1$ is $x^T M y$.  It has been observed that such bilinear games can be solved efficiently when the strategy polytope has polynomially many constraints \cite{Charnes53,KMvS94}.  In each of the examples described above, it is shown how to 
%The approach, then, is to attempt to find a way to 
map strategies from the original games to appropriate payoff-equivalent bilinear games, in which strategies are choices of marginal probabilities from the original game.  If one can also map a strategy in the bilinear game back to the original game, then one has a polytime reduction to the (solved) problem of finding equilibria of the bilinear game.  In each of these prior works it is this latter step -- mapping back to the original game -- that is the most demanding; this generally requires a problem-specific way to convert a profile of marginals into a corresponding mixed strategy in the original game.
%((BRENDAN: Consider arguing that using this approach directly is tough with Colonel Blotto, and then put the footnote here.  How much do we want to say about the previous approach, and how ours is different?))

\subsection{Our Contribution}

We first show how to compute equilibria of the Colonel Blotto game. %using bilinear constructions.
%
%In this work we present a general approach for computing equilibria in zero-sum games, motivated by these bilinear constructions.
%
%In this work we present the first polynomial-time algorithm for computing equilibria of the Colonel Blotto game.  
Like the works described above, our method is to consider a payoff-equivalent bilinear game defined over a space of appropriately-selected marginals (in this case, the distribution of soldiers to a given battlefield).
However, unlike those works, we do not explicitly construct a game-specific mapping to and from a polynomially-sized bilinear game.  We instead use a more general reduction, based on the idea that it suffices to solve linear optimization queries over strategy profiles in a (potentially exponentially-sized) bilinear game.  In other words, equilibrium computation reduces to the problem of finding a strategy that optimizes a given linear function over its marginal components.  We apply our reduction to the Colonel Blotto game by showing how to solve these requisite optimization queries, which can be done via dynamic programming.

The reduction described above follows from a repeated application of the classic equivalence of separation and optimization \cite{GLS81}.  In more detail, we formulate the equilibrium conditions as an LP whose feasibility region is the intersection of two polytopes: the first corresponding to the set of strategies of player 1, and the second encoding payoff constraints for player 2.  To find a solution of the LP via Ellipsoid method, it suffices to design a separation oracle for each polytope. However, as we show, separation oracles for the second polytope reduce to (and from) separation oracles for the set of strategies of player 2.  It therefore suffices to design separation oracles for the polytope of strategies for each player, and for this it is enough to perform linear optimization over those polytopes \cite{GLS81}.  Finally, to convert back to an equilibrium of the original game, we make use of a result from combinatorial optimization: the solution of an LP with polynomially many variables can always be expressed as a mixed strategy with a polynomial-size support, and such a mixed strategy can be computed using the separation oracles described previously \cite{GLS81}.

The reduction described above is not specific to the Colonel Blotto game: it applies to any zero-sum game, and any payoff-equivalent bilinear form thereof.
To the best of our knowledge, this general reduction from equilibrium computation to linear optimization has not previously been stated explicitly, although it has been alluded to in the security games literature\footnote{Independently and in parallel with an earlier version of this work, Xu et al. \cite{XFJCDT14} implicitly used a similar idea to solve a class of Stackleberg security games.} and similar ideas have been used to compute correlated equilibria in compact games \cite{JL15}.
In particular, it is notable that 
%This structure applies ellipsoid method to different subproblems, though 
one requires only a single linear optimization oracle, over the set of pure strategies, to both find an equilibrium of the bilinear game and convert this to a mixed equilibrium in the original game.
%However the novelty of our approach lies in putting separate parts (possibly used before) together to come up with a cohesive structure. 
%The fact that we need to solve a unique separation oracle for different parts of our method 
%This reuse of a single optimization task makes this approach significantly simpler than prior methods based on bilinear games.  
%In addition our method is more general in the sense that there are instances (e.g. the Matching duel) for which, as far as we know, no previously known method is applicable.
%Beyond the Colonel Blotto game, 
We demonstrate the generality of this approach by considering 
%other 
notable examples of games to which it can be applied.  In each case, our approach either results in the first known polytime algorithm for computing equilibria, or else significantly simplifies prior analysis.
Finally, we note that our approach also extends to approximations: given the ability to approximately answer separation oracle queries to within any fixed error $\epsilon > 0$, one can compute a corresponding approximation to the equilibrium payoffs.  %((BRENDAN: make that last statement more precise -- is a PTAS necessary?))

\paragraph{Dueling Games}
In a {\em dueling game}, introduced by \citeauthor{duel11} \cite{duel11},
%(STOC 2011), 
two competitors each try to design an algorithm for an optimization problem with an element of uncertainty, and each player's payoff is the probability of obtaining a better solution.  This framework falls within a natural class of ranking or social context games \cite{AKT08,BFHS09}, in which players separately play a base game and then receive ultimate payoffs determined by both their own outcomes and the outcomes of others.
\citeauthor{duel11} argue that this class of games models a variety of scenarios of competitions between algorithm designers: for example, competition between search engines (who must choose how to rank search results), or competition between hiring managers (who must choose from a pool of candidates in the style of the secretary problem).

%
%For example, consider two search engines compete to propose a good search result
%
\citeauthor{duel11} \cite{duel11} show how to compute a Nash equilibrium for certain dueling games, by developing mappings to and from bilinear games with compact representations.  We extend their method, and show how to expand the class of dueling games for which equilibria can be efficiently computed.  As one particular example, we introduce and solve the \emph{matching duel}.  In this game, two players each select a matching in a weighted graph, and each player's payoff is the probability that a randomly selected node would have a higher-weight match in that player's matching than in the opponent's.  Notably, since the matching polytope does not have a compact representation \cite{matching14}, the original method of \cite{duel11} is not sufficient to find equilibria of this game.  
%our extension (to bilinear forms that are non-compact but admit separation oracles) is necessary to apply equilibria of the matching duel.  
We also illustrate that our approach admits a significantly simplified analysis for some other dueling games previously analyzed by \citeauthor{duel11}

%\noindent \textbf{General Lotto Game.}
\paragraph{General Lotto Game}
%((BRENDAN: Consider shortening this subsection.  Also: is it fair to say that it uses the same high-level approach as the other results?))
\citeauthor{H07}~\cite{H07} considers %the discrete version of Colonel Blotto and studies 
%two variants of the Colonel Blotto game, namely, the Colonel Lotto and the General Lotto games.
a variant of the Colonel Blotto game, namely the \emph{General Lotto} game.
%In the {\em Colonel Lotto} game, each player partitions its troops into $k$ corps, where $k$ is the number of battlefields. For each player, one of these $k$ corps will be chosen uniformly at random. The two chosen corps fight determines the winner of the game. In this situation, each player is interested in an action which maximizes the probability of winning.  The {\em General Lotto} game is a variant of the Colonel Lotto game, where 
In this game, each player chooses a distribution over non-negative real numbers, subject to the constraint that its expectation must equal a certain fixed value.  
%of a non-negative random variable with a given expectation; the expectation would be the number of troops divided by the number of battlefields. 
%
A value is then drawn from each player's chosen distribution; the players' payoffs are then functions of these values.  
%The payoff of each player 
%in the General Lotto game 
%is then a function %of his random variable and his opponent's random variable where both random variables are drawn independently from their distributions. 
What is interesting about this game is that there are infinitely many pure strategies, which complicates equilibrium computation.  Nevertheless, we show that our techniques can be applied to this class of games as well, yielding a polynomial-time algorithm for computing Nash equilibria.
It is worth mentioning that the General Lotto game is an important problem by itself, and its continuous variant has been well studied in the literature (see, for example, \cite{BC80,SP06,H07,D11}).

\paragraph{Subsequent Work} The algorithm proposed in this work is later improved and simplified in ~\cite{behnezhad2016faster}.

%\subsection{Our Results}\label{our_results}
\section{Results and Techniques}\label{our_results}

We present a general method for computing Nash equilibria of a broad class of zero-sum games.  Our approach is to reduce the problem of computing equilibria of a given game to the problem of optimizing linear functions over the space of strategies in a payoff-equivalent bilinear game.

Before presenting our general reduction, we will first illustrate our techniques by considering the Colonel Blotto game as a specific example.  In Section \ref{sec:results.blotto} we describe our approach in detail for the Colonel Blotto game, explaining the process by which equilibria can be computed.  Then in Section \ref{sec:results.general} we will present the general reduction.  Further applications of this technique are provided in Section \ref{sec:duel} (for dueling games) and Appendix \ref{sec:lotto} (for the General Lotto game).

\subsection{Colonel Blotto} 
\label{sec:results.blotto}
%Despite a few algorithms which find  equilibria of special cases of the Colonel Blotto game, the problem is still open in its general form. 
%The main result of this paper is to propose polynomial-time algorithms for finding an equilibrium of Colonel Blotto, Colonel Lotto, and General Lotto games.
Here, we propose a polynomial-time algorithm for finding an equilibrium of discrete Colonel Blotto 
%, Colonel Lotto, and General Lotto games in their general form. 
in its general form.  We allow the game to be {\em asymmetric} across both the battlefields and the players. 
A game is {\em asymmetric across the battlefields} when different battlefields  have different contributions to the outcome of the game, and a game is {\em asymmetric across the players} when two players have different number of troops.

In the Colonel Blotto game, two players $A$ and $B$ simultaneously distribute $a$ and $b$ troops, respectively, over $k$ battlefields.  A pure strategy of player $A$ is a $k$-partition $\px=\langle x_1, x_2, \dots, x_k \rangle$ where $\sum_{i=1}^k x_i = a$, and a pure strategy of player $B$ is a $k$-partition $\py=\langle y_1, y_2, \ldots, y_k \rangle$ where $\sum_{i=1}^k y_i = b$. Let $u^A_i(x_i, y_i)$ and $u^B_i(x_i, y_i)$ be the payoff of player $A$ and player $B$ from the $i$-th battlefield, respectively. Note that the payoff functions of the $i$-th battlefield, $u^A_i$ and $u^B_i$, have $(a+1)\times(b+1)$ entries. This means the size of input is $\Theta(kab)$. Since  Colonel Blotto is a zero-sum game, we have $u^A_i(x_i, y_i)=-u^B_i(x_i, y_i)$\footnote{Note that in the Colonel Blotto game if $u^A_i(x_i, y_i)$ is not necessarily  equal to $-u^B_i(x_i, y_i)$ then a special case of this game with two battlefields can model an arbitrary 2-person normal-form game and thus finding a Nash Equilibrium would be PPAD-complete.}.
Note that we do not need to put any constraint on the payoff functions, and our result works for all payoff functions.
We also represent the total payoff of player $A$ and player $B$ by $h^A_{\mathcal{B}}(\px, \py) = \sum_i u^A_i(x_i, y_i)$ and $h^B_{\mathcal{B}}(\px, \py) = \sum_i u^B_i(x_i, y_i)$, respectively.
A mixed strategy of each player would be a probability distribution over his pure strategies.
\begin{theorem*} %[Theorems \ref{thm:salam} and \ref{thm:ref}]
\label{thm:blotto}
One can compute an equilibrium of any Colonel Blotto game in polynomial time.
\end{theorem*}
%Here, we give an intuition for proving the above theorem. The formal proof is in Appendix \ref{sec:blotto}%\footnote{The problem of finding a Nash equilibrium of the Colonel Blotto game when $u^A_i(x_i, y_i)$ may differ from $-u^B_i(x_i, y_i)$ is as hard as the problem of finding a Nash equilibrium of an arbitrary two players game. The hardness result also works in a special case when both colonels have exactly one troop.}. 
%
%\noindent \textit{Structure of the game.} 
\begin{proof}
Let $\X$ and $\Y$ be the set of all pure strategies of players $A$ and $B$ respectively, i.e., each member of $\X$ is a $k$-partition of $a$ troops and each member of $\Y$ is a $k$-partition of $b$ troops. We represent a mixed strategy of player $A$ with function $p:\X \rightarrow [0,1]$ such that $\sum_{\px \in \X} p(\px) =1$. Similarly, let function $q:\Y \rightarrow [0,1]$ be a mixed strategy of player $B$. We may also use $\x$ and $\y$, instead of $p$ and $q$, for referring to a mixed strategy of player $A$ and $B$ respectively.
Since Colonel Blotto is a zero-sum game, we leverage the MinMax theorem for finding an NE of the game. This theorem says that pair $(p^*, q^*)$ is an NE of the Colonel Blotto game if and only if strategies $p^*$ and $q^*$ maximize the guaranteed payoff of players $A$ and $B$ respectively \cite{minmax}. Now, we are going to find strategy $p^*$ of player $A$ which maximizes his guaranteed payoff. The same technique can be used for finding $q^*$.
%We can use the same technique for finding $q^*$. 
%
It is known that for each mixed strategy $p$,  at least one of the best-response strategies to $p$ is a pure strategy. Therefore, a solution to the following program characterizes strategy $p^*$.
\begin{eqnarray}
\label{Prog:2}
\max &U\\
s.t. &\sum_{\px \in \X} p_{\px} = 1,& \nonumber \\
&\sum_{\px \in \X} p_{\px}h^A_{\B}(\px, \py) \geq U, & \forall \py \in \Y, \nonumber
\end{eqnarray}
%Consider an optimum solution $(p^*, \py)$ of LP \ref{Prog:2}. Then $p^*$ would be a mixed strategy of player $A$ which maximizes his guaranteed payoff, and $\py$ would be a best-response strategy of player $B$ to mixed strategy $p^*$. 
Unfortunately, LP \ref{Prog:2} has $|\X|$ variables and $|\Y|+1$ constraints where $|\X|$ and $|\Y|+1$ are exponential.
We therefore cannot solve LP \ref{Prog:2} directly.
%Nevertheless, finding a solution to LP \ref{Prog:2} is problematic, as this linear program has an exponential number of variables and an exponential number of constraints. %One may hope to use the ellipsoid method and design an separation oracle which runs in polynomial-time for finding a solution of LP \ref{Prog:2}. 

\noindent \textit{Step 1: Transferring to a new space.} 
We address this issue by 
%This means it is not possible to solve the problem in this form, and transferring the solution space is essential 
transforming the solution space to a new space 
in which an LP equivalent to LP \ref{Prog:2} becomes tractable (See, e.g., \cite{azar2003}, for similar technique).
This new space will project mixed strategies onto the marginal probabilities for each (battlefield, troop count) pair.
%such that designing a separation oracle becomes tractable. 
For each pure strategy $\px \in \X$ of player $A$, we map it to a point in $\{0,1\}^{\n}$ where $\n = k \times (a+1)$. For convenience, we may abuse the notation, and index each point $\xh \in \{0,1\}^{\n}$ by two indices $i$ and $j$ such that $\hat{x}_{i,j}$ represents $\hat{x}_{(i-1)(a+1)+j+1}$. %We may use both $\hat{z}_{i,t}$ and $\hat{z}_j$ interchangeably. 
Now we map a pure strategy $\px$ to $\G(\px) = \xh \in \{0,1\}^{\n}$ such that $\hat{x}_{i,j} = 1$ if and only if $x_i = j$. In other words, if player $A$ puts $j$ troops in the $i$-th battlefield then $\hat{x}_{i, j}=1$. 
Let $\Ia=\{\xh \in \{0,1\}^{\n} | \exists \px \in \X, \G(\px) = \xh \}$ be the set of points in $\{0,1\}^{\n}$ which represent pure strategies  of player $A$.
Let $\MX$ and $\MY$ be the set of mixed strategies of players $A$ and $B$ respectively. Similarly, we map mixed strategy $\x$ to point $\G(\x) = \xh \in [0,1]^{\n}$ such that $\hat{x}_{i,j}$ represents the probability that mixed strategy $\x$ puts $j$ troops in the $i$-th battlefield. Note that mapping $\G$ is not necessarily one-to-one nor onto, i.e., each point in $[0,1]^{\n}$ may be mapped to zero, one, or more than one strategies. Let $\Sa=\{\xh \in [0,1]^{\n} | \exists \x \in \MX, \G(\x) = \xh \}$ be the set of points in $[0,1]^{\n}$ which represent at least one mixed strategy of player $A$. Similarly, we use function $\Gb$ to map each strategy of player $B$ to a point in $[0,1]^{\nb}$ where $\nb = k \times (b+1)$, and define $\Ib=\{\yh \in \{0,1\}^{\nb} | \exists \py \in \Y, \Gb(\py) = \yh \}$ and $\Sb=\{\yh \in [0,1]^{\nb} | \exists \y \in \MY, \Gb(\y) = \yh \}$. 
\begin{lemma*}
\label{polyhedron_new}
Set $\sets$ forms a convex polyhedron with an exponential number of vertices and facets.
\end{lemma*}

Now, we are ready to rewrite linear program \ref{Prog:2} in the new space as follows. 
\begin{eqnarray}
\label{Prog:3}
\max &U& \\
s.t. &\xh \in \Sa &\text{(Membership constraint)}\nonumber\\%\label{const:1}\\
&h_{\B}^A(\xh, \yh) \geq U, \text{\ \ \ }\forall \yh \in \Ib&\text{(Payoff constraints)}\nonumber %\label{const:2}
\end{eqnarray}
\iffalse
\begin{eqnarray}
\label{Prog:3}
\max &U& \\
s.t. &\xh \in \Sa &\text{(Membership constraint)}\label{const:1}\\
\sum_i \sum_{t_a} \sum_{t_b} &\hat{x}_{i, t_a} \hat{y}_{i, t_b} u_i^A(t_a, t_b) \geq U, \forall \yh \in \Ib,&\text{(Payoff constraints)} \label{const:2}
\end{eqnarray}
\fi
where $$h_{\B}^A(\xh, \yh) = \sum_{i=1}^k \sum_{t_a=0}^a \sum_{t_b=0}^b \hat{x}_{i, t_a} \hat{y}_{i, t_b} u_i^A(t_a, t_b)$$ is the expected payoff of player $A$.

\noindent \textit{Step 2: Solving LP \ref{Prog:3}.}  The modified LP above, LP \ref{Prog:3}, has exponentially many constraints, but only polynomially many variables. One can therefore apply the Ellipsoid method to solve the LP, given a separation oracle that runs in polynomial time \cite{lp1,lp2}.  By the equivalence of separation and optimization \cite{GLS81}, one can implement such a separation oracle given the ability to optimize linear functions over the polytopes $\Sa$ (for the membership constraints) and $\Sb$ (for the payoff constraints).

Stated more explicitly, given a sequence $c_0, c_1,\ldots, c_{k(m+1)}$, where $k$ is the number of battlefields and $m$ is the number of troops for a player,  the required oracle must find a pure strategy $x=(x_1, x_2, \dots, x_k)\in \pureset$ such that $\sum_{i=1}^kx_i=m$, and $\xh=\mathcal{G}(x)$ minimizes the following equation:
\begin{equation}\label{minpure_eq}
%c_0+c_1 \hat{x}_1 + \ldots+ c_n \hat{x}_n
c_0+\sum_{i=1}^{k(m+1)} c_i \hat{x}_i,
\end{equation}
and similarly for polytope $\Y$.
%Then we leverage this algorithm and design polynomial-time algorithms for the hyperplane separating oracle and best-response separating oracle. 
The following lemma shows that one can indeed find a minimizer of Equation \eqref{minpure_eq} in polynomial time.
%Algorithm \ref{alg:minpure} (\minpure) finds the minimizer of Equation \ref{minpure_eq}.
%\input{minpure_alg}

\begin{lemma}\label{lem:minpure}
Given two integers $m$ and $k$ and a sequence $c_0,c_1,\ldots,c_{k(m+1)}$, 
%algorithm \minpure\ correctly 
one can find (in polynomial time) an optimal pure strategy $x=(x_1, x_2, \dots, x_k)$ where $\sum_{i=1}^{k}x_i=m$, $\xh=\mathcal{G}(x)$ and $\xh$ minimizes $c_0+\sum_{i=1}^{k(m+1)} c_i \hat{x}_i$. %such that Equation \ref{minpure_eq} is minimized.
\end{lemma}
\begin{proof}
%Algorithm \minpure uses a dynamic programming approach to solve that problem.
%In Algorithm \minpure, 
We employ a dynamic programming approach.  Define $d[i,t]$ to be the minimum possible value of $c_0+\sum_{i'=1}^{i(t+1)} c_{i'}\hat{x}_{i'}$ where $\sum_{i'=1}^{i}x_{i'}=t$. Hence, $d[k,m]$ denotes the minimum possible value of $c_0+\sum_{i=1}^{k(m+1)} c_i \hat{x}_i$.  
%Now, we show that Algorithm \minpure\ correctly computes $d[i,t]$ for all $0\leq i \leq k$ and $0\leq t\leq m$. 
We have that $d[0,j]$ is equal to $c_0$ for all $j$.  For an arbitrary $i>0$ and $t$, the optimal strategy $x$ puts $0\leq t'\leq t$ units in the $i$-th battlefield and the applied cost in the equation \ref{minpure_eq} is equal to $c_{(i-1)(m+1)+t'+1}$. Thus, we can express $d[i,t]$ as
\begin{equation*}
d[i,t] = \min_{0\leq t'\leq t}\{d[i-1,t-t']+c_{(i-1)(m+1)+t'+1}\}.
\end{equation*}
Solving this dynamic program, we can find the allocation of troops that minimizes $\sum \alpha_{i} \xh_{i}$ in polynomial time, as required. 
%
%Since we have reduced this game to a polynomially-separable bilinear duel,  we can find an NE in polynomial time.
%Tracking the argmin for each entry $d[i,t]$ and considering the entry $d[k,m]$ then yields the optimal pure strategy.
%
%To compute the optimal pure strategy $x=(x_1,x_2,\ldots,x_k)$ we also keep a value 
%\begin{equation*}
%r[i,t] = \argmin_{0\leq t'\leq t}\{d[i-1,t-t']+c_{(i-1)(m+1)+t'+1}\}
%\end{equation*}
%which determines the number of units the optimal strategy should put in the $i$-th battlefield to minimize $c_0+\sum_{i'=1}^{i(t+1)} c_{i'}\hat{x}_{i'}$. Assuming we have correctly computed $x_{i+1},\ldots,x_k$, in line $16$, algorithm \minpure\ correctly computes $x_i$ which is equal to $r[i,m-\sum_{j=i+1}^{k} x_j]$. Since $x_k=r[k,m]$ we can conclude algorithm \minpure\ correctly computes the optimal strategy $x=(x_1, x_2, \dots, x_k)$.
%We define subproblem $\minpure(i,t)$ to be the problem of finding $\hx_1,\hx_2,\ldots,\hx_{j(a+1)}$ such that  $c_0+\sum_{i=1}^{j(a+1)} c_i\hat{x}_i$ is minimized.
\end{proof}

\noindent \textit{Step 3: Transferring to the original space.}
At last we should transfer the solution of LP \ref{Prog:3} to the original space. In particular, we are given a point $\xh \in \Sa$ and our goal is to find a strategy $\x \in \MX$ such that $\G(\x) = \xh$.  To achieve this, we invoke a classic result of \cite{GLS81} which states that an interior point of an $n$-dimensional polytope $P$ can be decomposed as a convex combination of at most $n+1$ extreme points of $P$, in polynomial time, given an oracle that optimizes linear functions over $P$.  Note that this is precisely the oracle required for Step 2, above.  Applying this result to the solution of LP \ref{Prog:3} in polytope $\Sa$, we obtain a convex decomposition of $\xh$ into extreme points of $\Sa$, say $\xh = \sum_i \alpha_i \xh_i$.  Since each $\xh_i$ corresponds to a pure strategy in $\X$, it is trivial to find point $\x_i$ with $\G(\x_i) = \xh_i$, since the marginals of each $\xh_i$ lie in $\{0,1\}$.  We then have that $\x = \sum_i \alpha_i \x_i$ is the required mixed strategy profile.

Combining these three steps, we find a Nash Equilibrium of the Colonel Blotto game in polynomial time, completing the proof of Theorem \ref{thm:blotto}. See Appendix \ref{sec:blotto} for more details.
\end{proof}

%We write another LP with an exponential number of variables for solving this problem. Then we design an appropriate separation oracle, and find a solution to the dual of the corresponding LP. See Appendix \ref{sec:ret} for details. The description of the proposed separation oracle is in Appendix \ref{oracle:sep}; in particular, it again uses the oracle from Step 2.1 as a subroutine. 

%\noindent \textbf{A general framework for bilinear games.}
\subsection{A general framework for bilinear games} 
\label{sec:results.general}
In our method for finding a Nash Equilibrium of the Colonel Blotto game, the main steps were to express the game as a bilinear game of polynomial dimension, solve for an equilibrium of the bilinear game, then express that point as an equilibrium of the original game.  To implement the final two steps, it sufficed to show how to optimize linear functions over the polytope of strategies in the bilinear game.
%each of the oracles required was implemented using \hporacle.  
This suggests a general reduction, where the equilibrium computation problem is reduced to finding the appropriate bilinear game and implementing the required optimization algorithm.  %\hporacle for a given game.  
%Indeed,
%we will leverage our method for finding an NE of Colonel Blotto game and design a general framework for computing an NE of a broad class of games. 
In other words, the method for computing Nash equilibria applies to a zero-sum game when:
\begin{enumerate}\itemsep -2pt % reduce space between items
\item One can transfer each strategy $x$ of player $A$ to $\G(x)=\xh \in R^{\n}$, and each strategy $y$ of player $B$ to $\Gb(y)=\yh \in R^{\nb}$ such that the payoff of the game for strategies $\xh$ and $\yh$ can be represented in a bilinear form based on $\xh$ and $\yh$, i.e., the payoff is $\xh^tM\yh$ where $M$ is a $\n \times \nb$ matrix.

%\begin{itemize}
%\item $\n$ and $\nb$ are polynomial.
%\item The set of all strategies of player $A$ and player $B$ in the transferred space, $\Sa$ and $\Sb$ respectively, are convex.
%\end{itemize}

\item For any given vector $\alpha$ and real number $\alpha_0$ we can find, in polynomial time, whether there is a pure strategy $\xh$ in the transferred space such that $\alpha_0 +  \sum_i\alpha_i\hat{x}_i \geq 0$.

\end{enumerate} 
%We call such game a \textit{polynomially-separable bilinear dueling game}. 
We refer to such a game as \emph{polynomially separable}.  A direct extension of the proof of Theorem \ref{thm:blotto} implies that Nash equilibria can be found for polynomially separable games.
%In Section \ref{sec:duel}, we design a polynomial-time algorithm which computes an NE of a zero-sum game with the above properties. The algorithm is similar to the proposed algorithm for the Colonel Blotto game in its nature. 

\begin{theorem*}
\label{thm:bi-sep}
There is a polytime algorithm which finds a Nash Equilibrium of a given polynomially separable game.
\end{theorem*}

%\textbf{[BJL: Should we say anything about being able to map strategies back to the original polytope, even when it isn't trivial to invert the mapping on extreme vertices?  Is our original proof actually more general than the simple proof given above, using GLS?]}

This general methodology can be used for finding a NE in many zero-sum games.
In subsequent sections, we show how our framework can be used to find Nash equilibria for a generalization of Blotto games, known as  games, and for a class of dueling games introduced by Immorlica et al. \cite{duel11}.

%\textbf{BJL: Removed some content here referring to dueling games specifically.  If we want to keep this content, we should move it to the section on dueling games.}

%In particular, we show how our framework can be used to find Nash equilibria of \textit{polynomially-separable bilinear dueling game}s, a generalization to the class of dueling games introduced by Immorlica et al. \cite{duel11}. %(STOC 2011). 
%
%It is worth mentioning that our approach does not need any kind of rounding. In particular, we present a general approach for polynomially-separable games which rounds any point in the solution space, and thus there is no need for a problem-specific rounding approach. 
%%
%Moreover, we show it is more convenient to prove a game is polynomially separable than to prove it is both polynomially representable and decomposable (compare results of \citet{duel11} and Section \ref{sec:duel}).

%\textbf{[ BJL:  Double-check the approximation result! ]}

We also show one can use similar techniques to compute the approximate equilibrium payoffs of a dueling game when we are not able to answer the separation problem in polynomial time but instead we can polynomially solve the $\epsilon$-separation problem for any $\epsilon > 0$.  The proof of Theorem \ref{thm:bi-sep} is deferred to Appendix \ref{sec:approx}.

\begin{theorem*}
\label{thm:bi-sep-approx}
Given an oracle function for the $\epsilon$-separation problem, one can find an $\epsilon$-approximation to the equilibrium payoffs of a polynomially separable game in polynomial time.
\end{theorem*}

\subsection{General Lotto}
The General Lotto game is a relaxation of the Colonel Lotto game (See \cite{H07} for details). In this game each player's strategy is a distribution of a nonnegative integer-valued random variable with a given expectation. In particular, players $A$ and $B$ simultaneously determine (two distributions of) two nonnegative integer-valued random variables $X$ and $Y$, respectively, such that $\E[X]=a$ and $\E[Y]=b$. The payoff of player $A$ is 
\begin{eqnarray}
\label{general_h}
h^A_{\Gamma}(X,Y) = \sum_{i=0}^{\infty} \sum_{j=0}^{\infty} \Pr(X=i) \Pr(Y=j)u(i,j),
\end{eqnarray}
and again the payoff of player $B$ is the negative of the payoff of player $A$, i.e., $h^B_{\Gamma}(X,Y) = -h^A_{\Gamma}(X,Y)$.
Hart \cite{H07} presents a solution for the General Lotto game when $u(i,j) = \sign(i-j)$.
%\footnote{Note that $\sign(i-j) = \begin{cases}1 & \mbox{if } i>j\\ -1 & \mbox{if } i<j\\ 0 & \mbox{if } i=j\end{cases} $}. 
Here, we generalize this result and present a polynomial-time algorithm for finding an equilibrium when $u$ is a {\em \distance\ function}. Function $u$ is a {\em \distance\ function}, if one can write it as $u(i,j) = f_u(i-j)$ such that $f_u$ is a monotone function and reaches its maximum value at $\um = f_u(\ut)$ where $\ut \in O(\poly(a,b))$. Note that $u(i,j)=\sign(i-j)$ is a \distance\ function where it reaches its maximum value at $i-j=1$. Now, we are ready to present our main result regarding the General Lotto game. 
\begin{theorem*} %[Theorem \ref{thm:dist}]
	%\label{thm:Generallotto}
	There is a polynomial-time algorithm which finds an equilibrium of a General Lotto game where the payoff function is a \distance\ function.
\end{theorem*}
\noindent \textbf{Main challenge}
Note that in the General Lotto game, each player has infinite number of pure strategies, and thus one cannot use neither our proposed algorithm for the Colonel Blotto game nor the technique of \cite{duel11} for solving the problem. We should prune strategies such that the problem becomes tractable. Therefore, we characterize the extreme point of the polytope of all strategies, and use this characterization for pruning possible strategies.

To the best of our knowledge, our algorithm is the first algorithm of this kind which computes an NE of a game with infinite number of pure strategies.
\section{Application to Dueling Games}
\label{sec:duel}
Immorlica et al. \cite{duel11} introduced the class of dueling games.
%which includes a wide variety of zero-sum games called dueling games. 
In these games, an optimization problem with an element of uncertainty is considered as a competition between two players. They also provide a technique for finding Nash equilibria for a set of games in this class. Later Dehghani \textit{et al.}\cite{dehghani2016price} studied dueling games in a non-computational point of view and proved upper bounds on the price of anarchy of many dueling games.  In this section, we formally define the dueling games and bilinear duels. Then, in Section \ref{our_method}, we describe our method and show that our technique solves a more general class of dueling games. Furthermore, we provide examples to show how our method can be a simpler tool for solving bilinear duel games compared to \cite{duel11} method.
%Moreover, our method provides a much simpler tool for solving such games. For example, we provide much simpler approaches to find a Nash equilibrium in two dueling games mentioned in \cite{duel11}. 
Finally, in Section \ref{subsec:matching}, we examine the matching duel game to provide an example where the method of \citet{duel11} does not work, but our presented method can yet be applied.
%On the other hand, to the best of our knowledge, the method provided in \cite{duel11} is not applicable for the Colonel Blotto game. 
\subsection{Dueling games}
%Many optimization problems can be defined in the form of a single player game. For example, consider a problem in which we are given a set of numbers each having a probability to be searched. We are to build a binary search tree consisting of the numbers that minimizes the expected time of finding the search query. The optimum solution can be easily found in polynomial time with a dynamic program. This problem can be modeled with a single player game where the set of strategies is the set of all possible binary search trees and the cost function is the expected query time. Now consider a two player game where each player is going to build a binary search tree and for a given query the winner is the player that can answer the query faster i.e. the depth of the searched element is lower in her binary search tree. Therefore, players intend to maximize the expected number of times they answer the query faster and beat the other player, therefore it is called a dueling game. 
Formally, dueling games are two player zero-sum games with a set of strategies $X$, a set of possible situations $\Omega$, a probability distribution $p$ over $\Omega$, and a cost function $c:X \times \Omega \rightarrow \mathbb{R}$ that defines the cost measure for each player based on her strategy and the element of uncertainty. 
%For example, in the binary search dueling game, $X$ is the set of all binary search trees, $\Omega$ is the set of all number queries, and $c(t \in X, \omega \in \Omega)$ is the depth of $\omega$ in $t$. 
The payoff of each player is defined as the probability that she beats her opponent minus the probability that she is beaten. More precisely, the utility function is defined as
%payoff for player $A$ $h^A(x,y)$ where $x$ and $y$ are strategies of player $A$ and player $B$ is defined as
\begin{equation*}
h^A(x,y) = -h^B(x, y)= \Pr_{\omega \sim p}[c(x,\omega) < c(y,\omega)] - \Pr_{\omega \sim p}[c(x,\omega) > c(y,\omega)]
\end{equation*}
where $x$ and $y$ are strategies for player $A$ and $B$ respectively.
In the following there are two dueling games mentioned in \cite{duel11}.

\textbf{Binary search tree duel.} In the Binary search tree duel, there is a set of elements $\Omega$ and a probability distribution $p$ over $\Omega$. Each player is going to construct a binary search tree containing the elements of $\Omega$.
% and for each $\omega \in \Omega$ one player beats the other if $\omega$ has shorter path to the root in her tree than her opponent. 
Strategy $x$ beats strategy $y$ for element $\omega  \in \Omega$ if and only if the path from $\omega$ to the root in $x$ is shorter than the path from $\omega$ to the root in $y$.
Thus, the set of strategies $X$ is the set of all binary search trees with elements of $\Omega$, and $c(x, \omega)$ is defined to be the depth of element $\omega$ in strategy $x$.

\textbf{Ranking duel.} In the Ranking duel, there is a set of $m$ pages $\Omega$, and a probability distribution $p$ over $\Omega$, notifying the probability that each page is going to be searched. In the Ranking duel, two search engines compete against each other. Each search engine has to provide a permutation of these pages, and a player beats the other if page $\omega$ comes earlier in her permutation. Hence, set of strategies $X$ is all $m!$ permutations of the pages and for permutation $x=(x_1, x_2, \ldots, x_m)$ and page $\omega$, $c(x, \omega)=i$ iff $\omega=x_i$.

%\textbf{Hiring duel.} In the Hiring duel, there are two employers $A$ and $B$ who want to interview two sets of workers $W_A = \langle a_1, a_2, \ldots a_m \rangle$ and $W_B = \langle b_1, b_2, \ldots, b_m\rangle$ sorted by their value, i.e. $v(i) > v(j)$ if and only if $i>j$ where $v(i)$ is the value of $i$-th worker. The employers interview the workers in a random order $\sigma_A$ and $\sigma_B$. At each step the employer interviews a worker and understands its value, then she should either hire the worker and finish the process or reject her and try to hire another worker from the rest. Note that the employers do not have any information about the order $\sigma$. Formally, the set of strategies is $blah blah$, $\Omega$ is the set of all pairs of orders $(\sigma_a, \sigma_b)$ and for strategy $x$ and pair of orders $\omega=(\sigma_a, \sigma_b)$, $c(x, \omega)$ is the rank of the worker chosen by $x$. 

\subsection{Dueling games are Polynomially Separable}
Consider a dueling game in which each strategy $\hat{x}$ of player $A$ is an $n(A)$ dimensional point in Euclidean space. Let $S_A$ be the convex hull of these strategy points. Thus each point in $S_A$ is a mixed strategy of player $A$. Similarly define strategy $\hat{y}$, $n(B)$, and $S_B$ for player $B$. A dueling game is {\em bilinear} if utility function $h^A(\hat{x}, \hat{y})$ has the form $\hat{x}^tM\hat{y}$ where $M$ is an $n(A) \times n(B)$ matrix. Again for player $B$, we have $h^B(\hat{x}, \hat{y})=-h^A(\hat{x}, \hat{y})$.
\citeauthor{duel11} \cite{duel11} provide a method for finding an equilibrium of a class of bilinear games which is defined as follows:
\begin{definition*}
{Polynomially-representable bilinear dueling games:} A bilinear dueling game is polynomially representable if one can present the convex hull of strategies $S_A$ and $S_B$ with $m$ polynomial linear constraints, i.e. there are $m$ vectors $\{v_1, v_2, \ldots, v_m\}$ and $m$ real numbers $\{b_1, b_2, \ldots, b_m \}$ such that $S_A=\{ \hat{x} \in R^{n(A)} | \forall{i \in \{1, 2, \ldots, m\}}, v_i.\hat{x} \geq b_i \}$. Similarly $S_B=\{ \hat{y} \in R^{n(B)} | \forall{i \in \{1, 2, \ldots, m'\}}, v'_i.\hat{y} \geq b'_i \}$.
\end{definition*}

%Next, we define another class of bilinear games which is a superset of polynomially-representable bilinear dueling games. 
%\begin{definition*}
%{Polynomially-separable bilinear dueling games:} A bilinear dueling game is polynomially separable if for any vector $v$ and real number $b$, we can determine in polynomial time whether there exists a strategy point $\hat{x}$, such that $b + v.\hat{x} \geq 0$.
%\end{definition*}

%\xxx{Here, we should somehow convince the reader that our method finds NE for polynomially separable bilinear duels.}

%We believe that our method for finding NE for polynomially separable bilinear duels is a strong tool for finding NE for dueling games in general. 
%N. Immorlica et. al. proposed a nice algorithm for finding NE for polynomially representable bilinear duels and provided algorithms for finding NE in Binary search tree duel, Hiring duel, and Ranking duel with tricky reductions to polynomially representable bilinear duel.
In the following theorem, we show that every polynomially representable bilinear duel is also  polynomially separable, as defined in Section \ref{sec:results.general}. This implies that the general reduction described in Section \ref{sec:results.general} can be used to solve polynomially representable bilinear duels as well. 

\begin{theorem*}
\label{thm:rep_sep}
Every polynomially-representable bilinear duel is polynomially separable.
\end{theorem*}
\begin{proof}
Let $S_A$ and $S_B$ be the set of strategy points for player $A$ and player $B$, respectively. We show that if $S_A$ can be specified with polynomial number of linear constraints, then one could design an algorithm that finds out whether there exists a point $\hat{x} \in S_A$  such that $\alpha_0 + \sum \alpha_i \hat{x}_i \geq 0$. Let $\{(v_1,b_1),(v_2,b_2),\ldots,(v_m,b_m)\}$ be the set of constraints which specify $S_A$ where $v_i$ is a vector of size $n(A)$ and $b_i$ is a real number. We need to check if there exists a point satisfying both constraints in $\{(v_1,b_1),(v_2,b_2),\ldots,(v_m,b_m)\}$ and $\alpha_0 + \sum \alpha_i \hat{x}_i \geq 0$. Recall that $m$ is polynomial. Since all these constraints are linear, we can solve this feasibility problem by a LP in polynomial time.
The same argument holds for $S_B$, therefore every polynomially representable bilinear duel is polynomially separable as well.
\end{proof}

\subsection{A Simplified Argument for Ranking and Binary Search Duels}\label{our_method}
In this section, we revisit some examples of dueling games, and show how to use Theorem \ref{thm:bi-sep} to establish that they can be solved in polynomial time. 
%can be used to solve a class of dueling games by reducing them to polynomially-separable duels.
%Our method has two steps. 
The application of Theorem \ref{thm:bi-sep} as two main steps.
%This reduction has two main steps.
First, it is necessary to express the duel as a bilinear game: that is, one must transfer every strategy of the players to a point in $n(A)$ and $n(B)$ dimensional space, such the outcome of the game can be determined for two given strategy points with an $n(A) \times n(B)$ matrix $M$. 
%
%In the second step, we show that the bilinear game is polynomially separable. Therefore, we find an 
Second, one must implement an oracle that determines whether there exists a strategy point satisfying a given linear constraint.

To illustrate our method more precisely, we propose a polynomial-time algorithm for finding an NE for ranking and binary search tree dueling games in what follows.
\begin{theorem*}
\label{thm:rank}
There exists an algorithm that finds an NE of the Ranking duel in polynomial time.
\end{theorem*}
\begin{proof}
We transfer each strategy $x$ of player $A$  to point $\hat{x}$
% = \langle \hat{x}_{1,1}, \hat{x}_{1,2}, \ldots, \hat{x}_{1,m}, \hat{x}_{2,1},\hat{x}_{2,2},\ldots,\hat{x}_{m,m} \rangle$ 
in $\mathbb{R}^{m^2}$ 
where $\hat{x}_{i,j}$ denotes the probability that  $\omega_{i}$ stands at position $j$ in $x$. The outcome of the game is determined by the following equation
\begin{equation*}
\sum_{i = 1}^m \sum_{j = 1}^m \sum_{k = j+1}^m \hat{x}_{i,j} \hat{y}_{i,k} p(\omega_i) -
\sum_{i = 1}^m \sum_{j = 1}^m \sum_{k = 1}^{j-1} \hat{x}_{i,j} \hat{y}_{i,k} p(\omega_i)
\end{equation*}
Where $p(\omega_i)$ denotes the probability that $\omega_{i}$ is searched.

Here, we need to provide an oracle that determines whether there exists a strategy point for a player that satisfies a given linear constraint $\alpha_0 + \sum \alpha_{i,j} \hat{x}_{i,j} \geq 0$. Since each pure strategy is a matching between pages and indices, we can find the pure strategy that maximizes $\sum \alpha_{i,j} \hat{x}_{i,j}$ with the maximum weighted matching algorithm. Therefore, this query can be answered in polynomial time. Since we have reduced this game to a polynomially-separable bilinear duel,  we can find a Nash equilibrium in polynomial time.
\end{proof}

\begin{theorem*}
There exists an algorithm that finds an NE of the Binary search tree duel in polynomial time.
\end{theorem*}
\begin{proof}
Here we map each strategy $x$ to the point $\hat{x} = \langle \hat{x}_{1,1}, \hat{x}_{1,2}, \ldots, \hat{x}_{1,m}, \hat{x}_{2,1},\hat{x}_{2,2},\ldots,\hat{x}_{m,m} \rangle \in \mathbb{R}^{m^2}$ where $\hat{x}_{i,j}$ denotes the probability that  depth of the $i$-th element is equal to $j$. Therefore, the payoff of the game for strategies $\hat{x}$ and $\hat{y}$ is equal to
 \begin{equation*}
\sum_{i = 1}^m \sum_{j = 1}^m \sum_{k = j+1}^m \hat{x}_{i,j} \hat{y}_{i,k} p(\omega_i) -
\sum_{i = 1}^m \sum_{j = 1}^m \sum_{k = 1}^{j-1} \hat{x}_{i,j} \hat{y}_{i,k} p(\omega_i)
\end{equation*}
Where $p(\omega_i)$ denotes the probability that $i$-th element is searched.

Next, we need to provide an oracle that determines whether there exists a strategy point for a player that satisfies a given linear constraint $\alpha_0 + \sum \alpha_{i,j} \hat{x}_{i,j} \geq 0$. To do this, we find the binary search tree that maximizes $\sum \alpha_{i,j} \hat{x}_{i,j}$. This can be done with a dynamic program. Let $D(a,b,k)$ denote the maximum value of $\sum_{i=a}^b \alpha_{i,j} \hat{x}_{i,j}$ for a subtree that its root is at depth $k$. $D(a,b,k)$ can be formulated as

%\begin{equation}
$$D(a,b,k) =
\begin{cases}
a<b, & \min_{a \leq c \leq b}\{ D(a,c-1,k+1) + D(c+1,b,k+1) + \alpha_{c,k} \}\\
a=b, & \alpha_{a,k}
\end{cases}$$
%\end{equation}
Therefore, we can find the Binary search tree which maximizes $\sum \alpha_{i,j} \hat{x}_{i,j}$ in polynomial time and see if it meets the constraint. Since we have reduced this game to a polynomially-separable bilinear duel,  we can find an NE in polynomial time.
\end{proof}

\subsection{Matching duel}
\label{subsec:matching}
In matching duel we are given a weighted graph $G=(V, E, W)$ which is not necessarily bipartite. In a matching duel each pure strategy of players is a perfect matching, set of possible situations $\Omega$ is the same as the set of nodes in $G$, and probability distribution $p$ over $\Omega$ determines the probability of selection of each node.
% and there is a probability distribution $p$ over the nodes of $G$, where for a node $u$, $p_u$ denotes the probability that $u$ will be selected. 
In this game, strategy $x$ beats strategy $y$ for element $\omega \in \Omega$ if $\omega$ is matched to a higher weighted edge in strategy $x$ than strategy $y$. 
%The winner of the game will be specified as follows. One node $u$ will be drawn randomly according to the probability distribution, and the winner is the one who has matched $u$ to a higher weighted edge.

The matching duel may find its application in a competition between websites that try to match people according to their desire. In this competition the website that suggest a better match for each user will get that user, and the goal of each website is to maximize the number of its users. 
We mention that the ranking duel is a special case of the matching duel, when $G$ is a complete bipartite graph with $n$ nodes on each side, in which the first part denotes the web pages and the second part denotes the positions in the ranking. Thus, the weight of the edge between page $i$ and rank $j$ is equal to $j$.

First, we describe how our method can solve this game and then we show the method of Immorlica et al. \cite{duel11} cannot be applied to find an NE of the matching duel.  
%Our method for the matching duel:
\begin{theorem*}
There exists an algorithm that finds an NE of the matching duel in polynomial time.
\end{theorem*}

\begin{proof}
We transfer every strategy $x$ to a point in $|E|$-dimensional Euclidean space $\hat{x}$, where $\hat{x}_e$ denotes the probability that $x$ chooses $e$ in the matching. 
Thus, the payoff function is bilinear and is as follows:
$$\sum_{\omega \in \Omega}\sum_{e_1 \in N(\omega)}\sum_{e_2 \in N(\omega)} [p(\omega) \hat{x}_{e_1} \hat{y}_{e_2} \times \text{sign}(w(e_1)-w(e_2))]$$
where $N(\omega)$ is the set of edges adjacent to $\omega$\footnote{Note that $\text{sign}(w)$ is $1$, $-1$, and $0$ if $w$ is positive, negative, and zero respectively.}.
Next, we need to prove that the game is polynomially separable. That is, 
%The remaining part is to provide a polynomial time algorithm for the separation problem, which is the following:
given a vector $\alpha$ and a real number $\alpha_0$, we are to find out whether there is a strategy $\hat{x}$ such that $\alpha_0 + \alpha . \hat{x} \geq 0$.
This problem can be solved by a maximum weighted prefect matching, where the graph is $G=(V, E, W)$ and $w(e)=\alpha_e$. Thus our framework can be used to find an NE of the matching duel in polynomial time.
\end{proof}
 
Note that Rothvoss \cite{matching14} showed that
%``The matching polytope has exponential extension complexity'' (Best paper of STOC 2014), 
the feasible strategy polytope (the perfect matching polytope) has exponentially many facets. Therefore, the prior approach represented in the work of Immorlica et al. \cite{duel11} is not applicable to the matching duel. This example shows that our framework nontrivially generalizes the method of Immorlica et al. \cite{duel11} and completes the presentation of our simpler and more powerful tool for solving bilinear duels.

{
\bibliography{blotto-abrv}
\bibliographystyle{acmsmall}
}

\appendix
\section{Colonel Blotto}\label{sec:blotto}

In this appendix we provide a more detailed desciption of our polynomial time algorithm for finding a Nash equilibrium (NE) of the Colonel Blotto game. Hart showed that the Colonel Lotto game is equivalent to a special case of the Colonel Blotto game \cite{H07}. Therefore, our algorithm could be used to find a NE of the Colonel Lotto game as well.

In section \ref{sec:transfer}, we present a procedure for mapping strategies of both players to a new space. The new space maintains the important information of each strategy and helps us to find a Nash equilibrium of the game.
Next in section \ref{sec:membership}, we show how we check the feasibility of the membership constraint in the new space. Moreover, we present a polynomial-time algorithm for determining an equilibrium of the Colonel Blotto game in the new space. At last in section \ref{sec:ret}, we present an algorithm which transfers a Nash equilibrium from the new space to the original space.

\subsection{Transferring to a new space}
\label{sec:transfer}
In this subsection we define a new notation for describing the strategies of players and discuss about the properties of the transferred strategies. Let $\na = k(a+1)$, and $\x$ be a strategy of player $A$. We define the function $\functiona$ in the following way:
$\functiona(\x) = \xh$ 
where
 $\xh$ is a point in $\Rna$ 
such that $\hat{x}_{(i-1)(a+1)+j+1}$ is equal to the probability that strategy $\x$ puts $j$ units in the $i$-th battlefield, for $1 \leq i \leq k$ and $0 \leq j \leq a$. For simplicity we may represent $\hat{x}_{(i-1)(a+1)+j+1}$ by $\hat{x}_{i,j}$.
We define $\nb$ and $\functionb$ similarly for player $B$. Let $n = max\{\na,\nb\}$.
Note that, $\functiona$ maps each strategy of the first player to exactly one point in $\Rna$. However, each point in $\Rna$ may be mapped to zero, one, or more than one strategies. Let us recall the definition of $\strategyset$ which is the set of all strategies of player $A$, and the definition of  $\sets$ which is
$$\sets = \{\xh \in [0,1]^{\na}| \exists \x \in \strategyset, \functiona(\x) = \xh \}.$$
%First, we show that $\sets$ forms a convex polyhedron in $\Rna$.\\
In order to design an algorithm for checking the membership constraint, we first demonstrate in Lemma \ref{polyhedron} that set $\sets$ is a polyhedron with an exponential number of vertices and facets. This lemma is a more formal statement of Lemma \ref{polyhedron_new}.
%
%Since we use Ellipsoid method to solve linear programs, it is important to 
Then we prove in Lemma \ref{lem:poly_constraint} that set $\sets$ can be formulated with $\cpx{2^{\poly(n)}}$ number of constraints. These results allow us to leverage the Ellipsoid method for checking  the membership constraint \cite{lp1}.
%\xxx[Hamid]{Ref}
%
\begin{lemma}
%\begin{theorem}
\label{polyhedron}
Set $\sets$ forms a convex polyhedron with no more than ${\na}^{\na}$ vertices
and no more than $\na^{(\na^2)}$ facets.
\end{lemma}
%\end{theorem}
See Appendix \ref{app:omitted} for the proof of Lemma \ref{polyhedron}.

\begin{lemma}
\label{lem:poly_constraint}
Set $\sets$ can be formulated with $\cpx{2^{\poly(n)}}$ number of constraints
\end{lemma}
See Appendix \ref{app:omitted} for the proof of Lemma \ref{lem:poly_constraint}.

\subsection{Checking the membership constraint and the payoff constraints}
\label{sec:membership}
As we briefly described in Subsection \ref{our_results}, the final goal of this section is to determine a NE of the Colonel Blotto game. To do this, we provide linear program \ref{Prog:3} and show that this LP can be solved in polynomial time. Since we use the Ellipsoid method to solve the LP, we have to implement an oracle function that reports a violating constraint for any infeasible solution. 
In this subsection we focus on the membership constraint of LP \ref{Prog:3} and show that for any infeasible point $\xh$ which violates membership constraint, a polynomial-time algorithm finds a hyperplane that separates $\xh$ from $\sets$.

%Now, we are ready to present a polynomial-time algorithm for checking the membership constraint.
% finds a hyperplane of $L$ which does not contain $Q$ or a hyperplane that separates $Q$ from $S$.\\
\begin{lemma}
\label{lem:membership}
There exists a polynomial time algorithm that gets a point $\xh$  as input, and either finds a hyperplane that separates $\xh$ from $\sets$, or reports that no such hyperplane exists.
\end{lemma}
\begin{proof}
Let $\ppointa{\xh}{\hat{x}}$. Consider the following LP, which we will refer to as LP \ref{Prog:separating_plane}:
\begin{eqnarray}
\label{Prog:separating_plane}
%\min. & & 0  \label{cn0}\\
& & \alpha_0 + \sum_{j=1}^{\na} \alpha_j\hat{x}_j < 0 \label{cn1}\\
& &  \alpha_0 + \sum_{j=1}^{\na} \alpha_j\hat{v_j} \geq 0 \label{cn2}\hspace*{2.12cm} \forall  \hat{v} \in \vs
%& &  \sum_{1 \leq i \leq n} \alpha_il_i = 0 \label{cn3}\hspace*{3.12cm} \forall l \in L
\end{eqnarray}
The variables of the linear program are $\alpha_0,\alpha_1,\ldots,\alpha_{\na}$, which describe the following hyperplane:
 $$\alpha_0+\sum_{j=1}^{\na}\alpha_j\hat{x}'_j = 0.$$
Constraints \ref{cn1} and \ref{cn2} force LP \ref{Prog:separating_plane} to find a hyperplane that separates $\xh$ from $\sets$. %Constraint \ref{cn3} states that the hyperplane should be perpendicular to all hyperplanes of $L$. 
%This means any feasible solution of LP represents a separating hyperplane.
%By Lemma \ref{totlem}, we know that  
Hence, LP \ref{Prog:separating_plane} finds a separating hyperplane if and only if $\xh$ is not in $\sets$.

\textit{Hyperplane separating oracle} is an oracle that gets variables $\alpha_0,\alpha_1,\ldots,\alpha_{\na}$ as input and finds if constraints \ref{cn2} are satisfied. Moreover, if some constraints are violated it returns at least one of the violated constraints. In Section \ref{oracles} we describe a polynomial-time algorithm for the hyperplane separating oracle. Constraint \ref{cn1} also can be checked in polynomial time.
% its algorithm and show that its running time is $\cpx{\poly(n)}$. There is only one other constraint and can be checked in polynomial time. The number of constraint of this
Our LP has $\na+1$ variables and $|\vs|+1$ constraints which is $\cpx{2^{\poly(n)}}$ by Lemma \ref{polyhedron}. Thus we can solve this LP in polynomial time with the Ellipsoid method \cite{lp1}. 
\end{proof}
In the next step, we present an algorithm to determine the outcome of the game when both players play optimally. We say $\x$ is an optimal strategy of player $A$, if it maximizes the guaranteed payoff of player $A$. By the MinMax Theorem, in a NE of a zero-sum game players play optimally \cite{minmax}. Therefore, it is enough to find an optimal strategy of both players.
Before we discuss the algorithm, we show the payoff $h_\blotto^A(\x,\y)$ can be determined by $\functiona(\x)$ and $\functionb(\y)$. Recall the definition of $h_\blotto^A(\xh,\yh)$ which is  $$ h_\blotto^A(\xh,\yh) = \sum_{i=1}^k \sum_{\alpha = 0}^a \sum_{\beta = 0}^b \hat{x}_{i,\alpha}\hat{y}_{i,\beta}u_i^A(\alpha,\beta).$$
\begin{lemma}
\label{formuli}
Let $\x \in \strategyset$ and $\y \in \strategysetb$ be two mixed strategies for player $A$ and $B$ respectively. Let $\xh = \functiona(\x)$ and $\yh = \functionb(\y)$. The outcome of the game is  determined by $h_\blotto^A(\xh,\yh)$. %the following equation: $$ h_\blotto^A(\xh,\yh) = \sum_{i=1}^k \sum_{\alpha = 0}^a \sum_{\betaas = 0}^b \hat{x}_{i,\alpha}\hat{y}_{i,\beta}u_i^A(\alpha,\beta)$$
\end{lemma}
\begin{proof}
Let $\x$ and $\y$ be two mixed strategies of players $A$ and $B$ respectively, and let $\mathbb{E}[u_i^A(\x,\y)]$ be the expected value of the outcome in battlefield $i$. We can write $\mathbb{E}[u_i^A(\x,\y)]$ as follows
$$\mathbb{E}[u_i^A(\x,\y)] = \sum_{\alpha=0}^a \sum_{\beta=0}^b \hat{x}_{i,\alpha} \hat{y}
_{i,\beta} u_i^A(\alpha,\beta).$$
We know that the total outcome of the game is the sum of the outcome in all battlefields, which is
$$
%h_\blotto^A(\x,\y) = 
\mathbb{E}\left[\sum_{i=1}^k u_i^A(\x,\y)\right] = \sum_{i=1}^k\mathbb{E}[u_i^A(\x,\y)] =
 \sum_{i=1}^k \sum_{\alpha = 0}^a \sum_{\beta = 0}^b \hat{x}_{i,\alpha}\hat{y}_{i,\beta}u_i^A(\alpha,\beta) = h_\blotto^A(\xh,\yh).$$
\end{proof}
%\qed

%By the definition, we have, 
%Note that, $\hat{x}_{i,\alpha}$ denotes the probability that first player puts $\alpha$ units in battlefield $i$ and $\hat{y}_{i,\beta}$ denotes the probability that player $B$ puts $\beta$ units in battlefield $i$. Hence, $h_\blotto^A(\x,\y)$ can be formulated as:

%Since the outcome $h_\blotto^A(\x,\y)$ can be determined with $\xh = \functiona(\x)$ and $\yh = \functionb(\y)$, we define a new notation $h_\blotto^A(\xh,\yh)$ which is equal to $h_\blotto^A(\x,\y)$ for $\x \in \strategyset$ and $\y \in \strategysetb$.\\
\begin{theorem}\label{thm:salam}
There exists a polynomial time algorithm that finds a NE of the Colonel Blotto game in the new space.
\end{theorem}
\begin{proof}
The Colonel Blotto is a zero-sum game, and the MinMax theorem states that a pair of strategies $(\xh, \yh)$ is a Nash equilibrium if $\xh$ and $\yh$ maximize the guaranteed payoff of players $A$ and $B$ respectively \cite{minmax}. 

Recall that LP  \ref{Prog:3} finds a point $\xh \in \sets$ which describes an optimal strategy of player $A$\footnote{The same procedure finds an optimal strategy of player $B$.}. 
This LP has $\na+1$ variables $\hat{x}_1,\hat{x}_2,\ldots,\hat{x}_{\na}$ and $U$ where $\hat{x}_1,\hat{x}_2,\ldots,\hat{x}_{\na}$ describe point $\xh$. The membership constraint guarantees $\xh$ is in $\sets$. 
It is known that in any normal-form game there always exists a best-response strategy which is a pure strategy \cite{game}. Hence, variable $U$ represents the maximum payoff of player $A$ with strategy $\xh$ when player $B$ plays his best-response strategy against $\xh$. 
Note that, Lemma \ref{formuli} shows $h_\blotto^A(\xh,\yh)$ is a linear function  of $\xh$, when $\yh$ is a fixed strategy of player $B$. This means the payoff constraints are linear constraints.
Putting all these together show, LP \ref{Prog:3}  finds a point $\xh$ such that:
\begin{enumerate}
\item There exists strategy $\x$ such that $\functiona(\x) = \xh$.
\item The minimum value of $h_\blotto^A(\x,\y)$ is maximized for every $\y \in \strategysetb$.
\end{enumerate}

Next we show that this LP can be solved in polynomial time with the Ellipsoid method. 
First, Lemma \ref{lem:membership} proposes a polynomial-time algorithm for checking the membership constraint.
Second, \textit{best-response separating oracle} is an oracle that gets point $\xh$ and variable $U$ as input and either reports point $\xh$ passes all payoff constrains or reports a violated payoff constraint. In Section \ref{oracles}, we will show that the running time of this oracle is $\cpx{\poly(n)}$. 

At last we prove LP \ref{Prog:3} has $\cpx{2^{\poly(n)}}$ number of constraints, and we can leverage the Ellipsoid method for finding a solution of this LP.
Note that Lemma \ref{lem:poly_constraint} indicates that set $\sets$ can be represented by $\cpx{2^{\poly(n)}}$ number of hyperplanes. On the other hand, Lemma \ref{polyhedron} shows LP \ref{Prog:3} has at most $|\Ib| =  \cpx{2^{\poly(n)}}$ constraints.
%Hence, for every potential solution of this LP we can find a violating constraint in polynomial time (if any).\\
%
%Lemma \ref{li4} states that hyperplanes of $C$ and $L$ are sufficient to separate points of $\Rna-\sets$ from $\sets$. Moreover, $|C|+|L| \in \cpx{2^{\poly(n)}}$. Hence, there exists a subset of constrains \ref{hhn2} with size $\cpx{2^{\poly(n)}}$ that form the feasibility region of $\xh$ and other constrains of \ref{hhn2} are redundant. Furthermore, number of \ref{hhn3} constraints is $\cpx{2^{\poly(n)}}$. Hence, there exists a subset of constraints with size $\cpx{2^{\poly(n)}}$ that specify the feasible region of the variables and there exists an algorithm that finds a violating constraint for every infeasible solution. Thus, LP can be solved in polynomial time with Ellipsoid method. 
%
%Note that this LP guarantees that there exists an optimal strategy $\x \in \strategyset$, such that $\functiona(\x) = \xh$ but it does not find $\x$. 
\end{proof}

\subsection{Finding a Nash equilibrium in the original space}
\label{sec:ret}
In the previous subsection, we presented an algorithm which finds a Nash equilibrium $(\functiona(\x), \functionb(\y))$ of the game in the new space. The remaining problem is to retrieve $\x$ from $\functiona(\x)$.
\begin{theorem}\label{thm:ref}
Given a point $\xh \in \sets$, there exists a polynomial time algorithm which finds a strategy $\x \in \strategyset$ such that $\functiona(\x) = \xh$.
\end{theorem}
\begin{proof}
Since every strategy of player $A$ is a convex combination of elements of $\pureset$, our goal is to find a feasible solution of the following LP.
\begin{eqnarray}
\min. & & 0 \label{cc0}\\
s.t. 
& &  \sum_{\px \in \pureset} \alpha_{\px} = 1 \label{cc1}\\
& &  \sum_{\px \in \pureset} \alpha_{\px} \functiona(\px)_j  = \hat{x}_j\hspace*{1.5cm} \forall  1 \leq j \leq \na \label{cc2}\\
& &  \alpha_{\px} \geq 0 \hspace*{3.7cm} \forall \px \in \pureset \label{cc3}
\end{eqnarray}
where $\alpha_{\px}$ is the probability of pure strategy $\px \in \pureset$.
Note that, this LP finds a mixed strategy of player $A$ by finding the probability of each pure strategy. Since every feasible solution is acceptable, objective function does not matter.
To find a solution of this LP, we write its dual LP as follows.
\begin{eqnarray}
\max. & & \beta_0 + \sum_{j=1}^{\na} \hat{x}_j \beta_j  \label{cp1} \\
s.t. & &  \beta_0 + \sum_{j=1}^{\na} \functiona(\px)_j \beta_j \leq 0 \hspace{1.5cm} \forall  \px \in \pureset \label{cp2}
\end{eqnarray}
Where variable $\beta_0$ stands for constraint \ref {cc1}, and variables $\beta_1,\beta_2,\ldots,\beta_{\na}$ stand for constraints  \ref{cc2}. An oracle similar to the \textit{hyperplane separating oracle} can find a violating constraint for any infeasible solution of the dual LP. Since the number of constraints in the dual LP is $|\Ia| = \cpx{2^{\poly(n)}}$ based on Lemma \ref{polyhedron}, we can use the Ellipsoid method and find an optimal solution of the dual LP in polynomial time.\\

\iffalse
The next challenge is to find an optimal solution of the primal LP from an optimal solution of the dual LP. The Ellipsoid method gives us an optimal dual solution. The complementary slackness theorem indicates that in an optimal solution of the primal LP, only the primal variables that correspond to the dual tight constraints may accept non-zero value. 
On the other hand, one can find polynomially many constraints of the dual which are tight at the optimal dual solution. 
Therefore, we can restrict the primal LP to the variables that correspond to the dual tight constraints and reduce the size of the primal LP. This leads us to an algorithm which finds an optimal solution of the primal LP in polynomial time. See \cite{lp1} for more details. \footnote{The same technique has been used in \cite{andersen2009interchanging} for finding an optimal primal solution from an optimal dual solution.} 
\fi

The next challenge is to find an optimal solution of the primal LP from an optimal solution of the dual LP. We resolve this problem by the following lemma. We know $\xh$ is in $\sets$.  This means there is strategy $\x \in \strategyset$ such that $\functiona(\x) = \xh$. Hence, linear program \ref{cc0} and its dual are feasible, and we can apply Lemma \ref{lemma:dualLP}.
\end{proof}

\begin{lemma}
\label{lemma:dualLP}
Assume we have a separation oracle for primal LP $\max\{c^Tx:Ax\leq b\}$ with exponentially many constraints and polynomially many variables. If primal LP is feasible, then there is a polynomial-time algorithm which returns an optimum solution of dual LP $\min\{b^Ty:A^Ty\geq c\}$.
\end{lemma}
\begin{proof}
Since the primal LP is feasible, we can assume $OPT=\max\{c^Tx:Ax\leq b\}$. The  Ellipsoid method returns an optimum solution of primal LP by doing binary search and finding the largest $K$ which guarantees feasibility of $\{c^Tx\leq K:Ax\leq b\}$. 
Let $(\hat{A},\hat{b})$ be the set of polynomially many constraints returned by the separation oracle during all iterations.
 We first prove $\max\{c^Tx:\hat{A}x\leq \hat{b}\} = OPT$. Note that $(\hat{A},\hat{b})$ is a set of constraints returned by the Ellipsoid method. Note that $(\hat{A},\hat{b})$ is a subset of all constraints $(A,b)$. This means every vector $x$ which satisfies $Ax \leq b$ will satisfy $\hat{A}x\leq \hat{b}$ as well. Therefore, $\max\{c^Tx:\hat{A}x\leq \hat{b}\} \geq \max\{c^Tx:Ax\leq b\} = OPT$. On the other hand, we know $(\hat{A},\hat{b})$ contains constraints which guarantees infeasibility of $\max\{c^Tx\geq OPT+\epsilon:Ax\leq b\}$. So, LP $\max\{c^Tx \geq OPI+\epsilon:\hat{A}x\leq \hat{b}\}$ is infeasible which means  $\max\{c^Tx:\hat{A}x\leq \hat{b}\}\leq OPT$. Putting all these together we can conclude that $\max\{c^Tx:\hat{A}x\leq \hat{b}\}=OPT$.

Linear program $\max\{c^Tx:\hat{A}x\leq \hat{b}\}$ has polynomially many constraints and polynomially many variables, and we can find an optimum solution to its dual $\min\{\hat{b}^T\hat{y}:\hat{A}^T\hat{y}\geq c\}$ in polynomial time. Let $\hat{y}^*$ be an optimum solution of dual LP $\min\{\hat{b}^T\hat{y}:\hat{A}^T\hat{y}\geq c\}$, 
and let $S=\{i| (A_i,b_i) \mbox{ is in } (\hat{A},\hat{b})\}$ where $A_i$ is the i-th row of matrix $A$, i.e.,  be set of indices corresponding to constraints in $(\hat{A}, \hat{b})$.
For every vector $y$ and every set of indices $R$ we define $y_R$ to be the projection of vector $y$ on set $R$.
Now consider vector $y^*$ as a solution of dual LP $\min\{b^Ty:A^Ty\geq c\}$ such that $y^*_S = \hat{y}^*$ and $y^*_i =0$ for all $i\not \in S$. We prove $y^*$ is an optimum solution of dual LP $\min\{b^Ty:A^Ty\geq c\}$ as follows:
\begin{itemize}
\item We first show $y^*$ is feasible. Note that $y^*_i=0$ for all $i \not \in S$ which means $A^T y^* = \hat{A}^T\hat{y}^* \geq c$ where the last inequality comes from the feasibility of $\hat{y}^*$ in dual LP $\min\{\hat{b}^Ty:\hat{A}^Ty\geq c\}$.
\item 
Note that $b^T y^* = \sum_i b^T_i y^*_i= \sum_{i \in S} b^T_i y^*_i + \sum_{i \not \in S} b^T_i y^*_i = \hat{b}^T \hat{y}^*$. The last equality comes from the facts that $y^*_i=0$ for all $i \not \in S$, and $\sum_{i \in S} b^T_i y^*_i = \hat{b}^T \hat{y}^*$. Since $\hat{y}^*$ is an optimum solution of dual the LP $\min\{\hat{b}^Ty:\hat{A}^Ty\geq c\}$, by the weak duality, it is equal to $\max\{c^Tx:\hat{A}x\leq \hat{b}\}=OPT$. Therefore, $b^T y^*= OPT$.
\end{itemize}
We have proved $y^*$ is a feasible solution to dual LP $\min\{b^Ty:A^Ty\geq c\}$ and $b^T y^* = OPT$. We also know $OPT=\max\{c^Tx:Ax\leq b\}$ by definition. Therefore, the weak duality insures $y^*$ is an optimum solution of dual LP $\min\{b^Ty:A^Ty\geq c\}$.
%\begin{equation}l 
%y^*_i = \begin{cases

\end{proof}

\section{General Lotto}
\label{sec:lotto}

In this section we study the General Lotto game. An instance of the General Lotto game is defined by $\Gamma(a, b, u)$, where players $A$ and $B$ simultaneously define  probability distributions of non-negative integers  $X$ and $Y$, respectively, such that $\mathbb{E}[X]=a$ and $\mathbb{E}[Y]=b$. In this game,  player $A$'s aim is to maximize $h_\Gamma^A(X,Y)$ and player $B$'s aim is to maximize $h_\Gamma^B(X,Y)=-h_\Gamma^A(X,Y)$ where $h_\Gamma^A(X,Y)$ is defined as $\mathbb{E}_{i \sim X, j \sim Y} u(i,j)$. 
The previous studies of the General Lotto game considered a special case of the problem where $u(i,j)=sign(i-j)$\footnote{
$sign(x)=\begin{cases}
-1 & x<0\\
0 & x=0\\
1 & x>0
\end{cases}$
}. Here, we generalize the payoff function to a  \distance\ function and present an algorithm for finding a Nash equilibrium of the General Lotto game in this case. 
Function $u$ is a {\distance\ function}, if one can write it as $u(i,j) = f_u(i-j)$ such that $f_u$ is a monotone function and reaches its maximum value at $f_u(t_u)$ where $t_u \in O(\poly(a,b))$. 

We first define a new version of the General Lotto game, which is called the \textit{finite General Lotto} game. We prove a Nash equilibrium of the finite General Lotto game can be found in polynomial time. Then we reduce the problem of finding a Nash equilibrium of the General Lotto game with a \distance\ function to the problem of finding a Nash equilibrium of the finite General Lotto game. This helps us to propose a polynomial-time algorithm which finds a Nash equilibrium of the General Lotto game where the payoff function is a \distance\ function.
\subsection{Finite General Lotto}
We define the finite General Lotto game $\Gamma(a,b,u,S)$ to be an instance of the General Lotto game where every strategy of players is a distribution over a finite set of numbers $S$. Here, we leverage our general technique to show the finite General Lotto game is a polynomially-separable bilinear game and, as a consequence, it leads to a polynomial time algorithm to find a Nash equilibrium for this game.

\begin{theorem}\label{finiteTheorem-a}
There exists an algorithm which finds a Nash equilibrium of the finite General Lotto game $\Gamma(a,b,u,S)$
	in time $O(poly(|S|))$.
\end{theorem}
\begin{proof}
First we map each strategy $X$ to a point $\hat{x}=\langle \hat{x}_1, \hat{x}_2, \ldots, \hat{x}_{|S|}\rangle$, where $\hat{x}_i$ denotes $\pr{X}{S_i}$. Without loss of generality we assume the elements of $S$ are sorted in strictly ascending order, i.e. for each $1 \leq i < j \leq |S|$, $S_i < S_j$. Now the utility of player $A$ when $A$ plays a strategy corresponding to $\hat{x}$ and $B$ plays a strategy corresponding to $\hat{y}$ is obtained by the following linear function.
$$h^A_{\Gamma}(\hat{x},\hat{y}) = \sum_{i=1}^{|S|} \sum_{j=1}^{i-1} \hat{x}\hat{y}u(i,j) - \sum_{i=1}^{|S|} \sum_{j=i+1}^{|S|} \hat{x}\hat{y}u(i,j).$$
Therefore the game is bilinear. Now we prove the game is polynomially separable.

	Given a real number $r$ and a vector $v$, we provide a polynomial-time algorithm which determines whether there exists a strategy point $\hat{x}$ such that $r+v.\hat{x} \geq 0$. We design the following feasibility program for this problem with $|S|$ variables $\hat{x}_1$ to $\hat{x}_{|S|}$ and three constraints.	
	\begin{align}
	&\sum_{i=1}^{|S|}\hat{x}_i=1 \label{lpc1-a}\\ 
	&\sum_{i=1}^{|S|}\hat{x}_iS_i=a \label{lpc2-a}\\ 
	&r+v.\hat{x}\geq 0 \label{lpc3-a}
	\end{align}
	Constraints \ref{lpc1-a} and \ref{lpc2-a} force the variables to represent a valid strategy point (i.e., the probabilities sum to 1 and the expectation equals $a$). Thus every point $\hat{x}$ is a valid strategy point \textit{iff} it satisfies Constraints \ref{lpc1-a} and \ref{lpc2-a}. On the other hand, Constraint \ref{lpc3-a} enforces the program to satisfy the given linear constraint of the separation problem. Thus there exists a strategy point $\hat{x}$ such that $b+v.\hat{x} \geq 0$ \textit{iff} there is a solution for the feasibility LP. The feasibility of the program can be determined in polynomial time, hence, the separation problem is polynomially tractable.

Therefore the finite General Lotto game is a polynomially-separable bilinear game and by Theorem \ref{thm:bi-sep}, there exists a polynomial-time algorithm which finds a Nash equilibrium of the finite General Lotto game.
\end{proof}

\subsection{General Lotto with \distance\ functions}
\label{sec:distance}
In this section, we consider General Lotto game $\Gamma(a, b, u)$ where $u$ is a \distance\ function and design a polynomial-time algorithm for finding a Nash equilibrium of the game. In the following part of this section we assume $a \leq b$. Recall the definition of \distance\ functions.
%First we define the distance functions as follows.
\begin{definition}
Function $u$ is a {\distance\ function}, if one can write it as $u(i,j) = f_u(i-j)$ such that $f_u$ is a monotone function and reaches its maximum value at $\um = f_u(\ut)$ where $\ut \in O(\poly(a,b))$. We call $\ut$ the threshold of function $u$, and $\um$ the maximum of function $u$ \footnote{Note that sign function is a special case of distance functions.} 
.
\end{definition}

First we define a notion of paired strategies and claim that for every strategy of a player there is a best-response strategy which is a paired strategy. Then, using this observation, 
we can prove nice bounds on the set of optimal strategies.
%we design a linear program which finds an optimal strategy of each player. This results in a polynomial-time algorithm for finding a Nash equilibrium of the given finite General Lotto game.

Consider a probability distribution which only allows two possible outcomes, i.e., there are only two elements in $S$ with non-zero probabilities. We call such a distribution a \textit{paired strategy}. We define $T_{i,j}$ to be a paired strategy which only has non-zero probabilities at elements $i$ and $j$. %where $i,j\in \mathbb{N}$. 
Furthermore, we define $T_{i,j}^a$ to be a paired strategy with $\E[T_{i,j}^a]=a$. In paired strategy $T_{i,j}^a$, the probabilities of elements $i$ and $j$ are determined by $\alpha_{i,j}^a=Pr(T_{i,j}^a=i) = \frac{a-j}{i-j}$ and $\alpha_{j,i}^a = Pr(T_{i,j}^a=j)= \frac{a-i}{j-i}=1-\alpha_{i,j}^a$. In the following structural lemma, we show that every distribution $T$ over a finite set $S$ can be constructed by a set of paired strategies.
%The following theorem reveals an important structural property of General Lotto games.

\begin{lemma}\label{pairLemma-a}
	For every distribution $T$ over $S$ with $\E[T] = a$ and $t$ elements with non-zero probability, there are $m \leq t$ paired strategies $ \sigma_1,\sigma_2,\ldots, \sigma_m$  such that $T=\sum_{r=1}^m \beta_r \sigma_r$\footnote{This lemma claims that the strategy of a player in the finite General Lotto game can be written as a probability distribution over paired strategies. Thus $T=\beta_1\sigma_1+\beta_2\sigma_2+\ldots+\beta_m\sigma_m$ describes a strategy in which the paired strategy $\sigma_i$ is played with probability $\beta_i$.}, and for all $1 \leq i \leq m$ we have $\beta_i\in[0,1]$ and $\E[\sigma_i]=a$.
\end{lemma}
\begin{proof}
	We prove this claim by induction on the number of elements with non-zero probabilities in $T$. If there is only one element with non-zero probability  in $T$, i.e., $t=1$, then we have $Pr(T=a)=1$. Thus $T=T_{0,a}^a$ is a paired strategy and the claim holds by setting $\sigma_1=T_{0,a}^a$ and $\beta_1=1$.
	%We present a constructive proof. Consider a distribution $T$ with $\mathbb{E}(T)=a$. If $Pr(T=a) = 1$, then $T=T_{0,a}^a$ is a paired strategy and the claim is correct (let $\sigma_1=T_{0,a}^a$ and $\beta_1=1$).
	Now assuming the claim holds for all $1\leq t'< t$, we prove the claim also holds for $t$. Suppose $T$ be a probability distribution with $t$ non-zero probability elements. Since $t\geq 2$, there should be some $i<a$ with $Pr(T=i)>0$ and some $j>a$ with $Pr(T=j)>0$. We choose the largest possible number $0<\beta\leq 1$ such that $\beta \alpha_{i,j}^a\leq Pr(T=i)$ and $\beta \alpha_{j,i}^a\leq Pr(T=j)$. If $\beta = 1$, then $Pr(T=i)+Pr(T=j)\geq \alpha_{i,j}^a + \alpha_{j,i}^a=1$. This means
	$T=T_{i,j}^a$ is a paired strategy and the claim holds by setting $\sigma_1=T_{i,j}^a$ and $\beta_1=1$. Otherwise, we can write $T=(1-\beta)T'+\beta T_{i,j}^a$ where $T'=\frac{T-\beta T_{i,j}^a}{1-\beta}$. Furthermore we have
	\begin{equation*}
	\E[T'] = \frac{\E[T-\beta T_{i,j}^a]}{1-\beta} =\frac{\E[T]-\beta E[T_{i,j}^a]}{1-\beta} = a.
	\end{equation*}
	We select $\beta$ such that at least one of the probabilities $Pr[T'=i]$ or $Pr[T'=j]$ becomes zero. Thus compared to $T$, the number of elements with non-zero probability in $T'$ is at least decreased by one, and by the induction hypothesis we can write $T'=\beta'_1 \sigma'_1+\beta'_2 \sigma'_2+\ldots+\beta'_{m'} \sigma'_{m'}$ where $m'\leq t-1$. Let $\beta_i=(1-\beta)\beta'_i$ and $\sigma_i=\sigma'_i$ for $1\leq i\leq m'$ and $\beta_{m'+1}=\beta$ and $\sigma_{m'+1}=T_{i,j}^a$. Now we can write $T=\beta_1 \sigma_1+\beta_2 \sigma_2+\ldots+\beta_{m'+1} \sigma_{m'+1}$ where each $\sigma_i$ is a paired strategy and $E(\sigma_i)=a$. Furthermore, $m=m'+1\leq t$ and the proof is complete. Since $t \leq |S|$, $m$ is polynomial in the size of input. Therefore paired strategies $\sigma_1,\sigma_2,\ldots,\sigma_m$ and their corresponding coefficients $\beta_1,\beta_2,\ldots,\beta_m$ can be computed in polynomial time.
\end{proof}

\begin{lemma}\label{bslemma-a}
	%For each strategy $X$ of player $A$ in finite General Lotto game $\Gamma(a,b,u,S)$, there is a best-response strategy $Y$ for player $B$ such that $Y=T_{i,j}^b$ for some $i,j\in S$.
	For every strategy of player $A$ in a finite General Lotto game there is a best-response strategy of player $B$ which is a paired strategy.
\end{lemma}

\begin{proof}
	Consider finite General Lotto game $\Gamma(a,b,u,S)$, strategy $X$ of player $A$, and a best-response strategy $Z$ of player $B$. Since $Z$ is a distribution on $S$, by using Lemma~\ref{pairLemma-a} we can write $Z=\sum_{r=1}^m\beta_r\sigma_r$. Thus, we have $h_\Gamma^B(X,Z)=h_\Gamma^B(X, \sum_{r=1}^m\beta_r \sigma_r)$ and because of the linearity of expectation we can write $h_\Gamma^B(X,Z)=\sum_{r=1}^m\beta_r h_\Gamma^B(X,\sigma_r)$. Since $Z$ is a best-response strategy, we have:
	%\begin{equation}
	%\forall 1\leq i,j\leq m, \hspace{.22cm} h_\Gamma^B(X,\sigma_i)= %h_\Gamma^B(X,\sigma_j).
	%\end{equation}
	\begin{equation}
	\forall 1\leq r\leq m, \hspace{.22cm} h_\Gamma^B(X,\sigma_r)= h_\Gamma^B(X,Z).
	\end{equation}
	This means paired strategy $\sigma_r$, for each $1 \leq r \leq m$, is a best-response strategy of player $B$.
	%Furthermore, since $\sum_{i=1}^m \beta_i = 1$ we have $h_\Gamma^B(X,Z)=h_\Gamma^B(X,\sigma_i)$ for all $1\leq i\leq m$. Thus, for every $1\leq i \leq m$ $Y=\sigma_i$ is a best-response strategy for player $B$.
\end{proof}

In the following lemmas, using the structural property of the best-response strategies, we show some bounds for each player's optimal strategies. 
\begin{lemma}\label{mins-a}
For any strategy $X$ with $\mathbb{E}[X]=c$ and any integer $j$ we have $\sum_{i=0}^j{\pr{X}{i}} \geq 1-\frac{c}{j+1}$.
\end{lemma}
\begin{proof}
Since $\sum_{i=0}^{+\infty}{i\pr{X}{i}} = c$, we have $(j+1)\sum_{i=j+1}^{+\infty}{\pr{X}{i}} \leq c$. This implies $$\sum_{i=0}^j{\pr{X}{i}} = 1-\sum_{i=j+1}^{+\infty}{\pr{X}{i}} \geq 1-\frac{c}{j+1}.$$
\end{proof}

\begin{lemma}\label{maxs-a}
Consider  Nash equilibrium $(X, Y)$ of General Lotto game $\Gamma(a, b, u)$ where $u$ is a \distance\ function with threshold $\ut$. We have $\sum_{i=0}^{a-1}{\pr{Y}{i}} \leq \frac{\ut}{\ut+1}$.
%\xxx[Hamid]{$a-1$ or $b-1$? I think it should be $a-1$.}
\end{lemma}
\begin{proof}
Let $X'$ be a pair distribution of player $A$ that chooses $a-1$ with probability $p$ and chooses $a+{\ut}-1$ with probability $1-p$. Thus $p=\frac{\ut-1}{\ut}$. The payoff of playing strategy $X'$ against $Y$ is 
\begin{equation}\label{ineq0}
  \begin{aligned}
		&h_{\Gamma}^A(X',Y)=\sum_{i=0}^{+\infty}\pr{Y}{i}[pu(a-1, i) + (1-p)u(a+{\ut}-1, i)].%=\\
%		& \sum_{i=0}^{a-1}{\pr{Y}{i}(pu(a-1, i)+(1-p)u(a+{\ut}-1, i))}+\\
%		& \sum_{i=a}^{+\infty}{\pr{Y}{i}(pu(a-1, i) + (1-p)u(a+{\ut}-1, i))}.
  \end{aligned}
\end{equation}

Note that by the definition of $u$, $u(i, j) \geq 0$ if and only if $i-j \geq 0$ and $u(i, j)\leq0$ if and only if $i-j \leq 0$. Furthermore, if $i-j \geq {\ut}$ then $u(i, j)={\um}$ and if $i-j \leq -{\ut}$ then $u(i, j)=-{\um}$. Therefore,
%\begin{eqnarray}
%\begin{aligned}
%	&\sum_{i=0}^{a-1}{\pr{Y}{i}pu(a-1, i)} \geq      0 \\
%	&\sum_{i=0}^{a-1}{\pr{Y}{i}(1-p)u(a+{\ut}-1, i)} \geq     (1-p){\um}\sum_{i=0}^{a-1}{\pr{Y}{i}} \\
%	&\sum_{i=a}^{+\infty}{\pr{Y}{i}(pu(a-1, i)+(1-p)u(a+{\ut}-1, i))} \geq     -{\um}\sum_{i=a}^{+\infty}{\pr{Y}{i}}
%\end{aligned}
%\end{eqnarray}

\begin{equation}\label{ineq1}
  \begin{aligned}
	&\sum_{i=0}^{a-1}{\pr{Y}{i}pu(a-1, i)} &\geq &     &0& \\
	&\sum_{i=0}^{a-1}{\pr{Y}{i}(1-p)u(a+{\ut}-1, i)} &\geq &    &&(1-p){\um}\sum_{i=0}^{a-1}{\pr{Y}{i}} \\
	&\sum_{i=a}^{+\infty}{\pr{Y}{i}[pu(a-1, i)+(1-p)u(a+{\ut}-1, i)]} &\geq &    &&-{\um}\sum_{i=a}^{+\infty}{\pr{Y}{i}}
%	&\sum_{i=a}^{+\infty}{\pr{Y}{i}(1-p)u(a+{\ut}-1, i)} &\geq &    &&-(1-p){\um}\sum_{i=a}^{+\infty}{\pr{Y}{i}}
  \end{aligned}
\end{equation}

Note that $a \leq b$ and $(X,Y)$ is a Nash equilibrium which means $h_{\Gamma}^A(X,Y) \leq 0$. This implies  $h_{\Gamma}^A(X',Y) \leq 0$. Thus, by applying  Equality \ref{ineq0} and Inequality \ref{ineq1} we have 
\begin{equation}\nonumber
\begin{aligned}
 0\geq h_\Gamma^A(X',Y) \geq     (1-p){\um}\sum_{i=0}^{a-1}{\pr{Y}{i}} 
 -{\um}\sum_{i=a}^{+\infty}{\pr{Y}{i}},
%     -p{\um}(1-\sum_{i=0}^{a-1}{\pr{Y}{i}})
%     -(1-	p){\um}(1-\sum_{i=0}^{a-1}{\pr{Y}{i}}). 
 \end{aligned}
 \end{equation}
which implies 
%\begin{equation*}
$\sum_{i=0}^{+\infty}{\pr{Y}{i}} \geq
  (2-p)\sum_{i=0}^{a-1}{\pr{Y}{i}}$. 
%\end{equation*}
  Thus, $ \sum_{i=0}^{a-1}{\pr{Y}{i}} \leq \frac{1}{2-p}$. By substituting $\frac{\ut-1}{\ut}$ instead of $p$ we can conclude $\sum_{i=0}^{a-1}{\pr{Y}{i}} \leq \frac{{\ut}}{{\ut}+1}$.
%  \qed
\end{proof}

In the following lemma we provide an upper-bound for the maximum variable with non-zero probability of a player's strategy in the equilibrium.
\begin{lemma}\label{finiteLemma-a}
Consider a Nash equilibrium $(X, Y)$ of General Lotto game $\Gamma(a, b, u)$ where $u$ is a \distance\ function with threshold $\ut$. If $\hat{u} = (4b{\ut}+4b+{\ut})(2{\ut}+2)$, then we have $Pr(Y>\hat{u}+{\ut})=0$ and $Pr(X>\hat{u}) = 0$.
%for any variable $x$ and $y $ with non-zero probability $x\leq 2(4a{\ut}+4a)+3{\ut}$ and $ y \leq 2(4a{\ut}+4a)+2{\ut}$, where $u$ is a distance function.
\end{lemma}
\begin{proof}
First, we prove for any integer $z>\hat{u}$, $\pr{X}{z} = 0$. The proof is by contradiction. Let $z > \hat{u}$ be an integer with non-zero probability in $X$. Thus there is an integer $x < a$ with non-zero probability in $X$. Consider the pair distribution $T_{x, z}^a$. We define another pair distribution $T_{x, y}^a$ where $y=4b{\ut}+4b+{\ut}$. 

Consider strategy $X^{\epsilon} = X - \epsilon T_{x,z}^a + \epsilon T_{x,y}^a$.
Note that $(X, Y)$ is a Nash equilibrium of the game. This means strategy $X$ is a best response of player $A$ to strategy $Y$ of player $B$ which implies $h_{\Gamma}^A(X, Y) \geq h_{\Gamma}^A(X^{\epsilon}, Y)$. On the other hand, because of the linearity of expectation we can write 
$$h_{\Gamma}^A(X^{\epsilon}, Y) = h_{\Gamma}^A(X, Y) - \epsilon h_{\Gamma}^A(T_{x,z}^a, Y) + \epsilon h_{\Gamma}^A(T_{x,y}^a, Y).$$
Therefore, we conclude $w=h_{\Gamma}^A(T_{x, z}^a, Y) - h_{\Gamma}^A(T_{x, y}^a, Y) \geq 0$.
Let $p=\alpha_{z, x}^a$ and $q=\alpha_{y, x}^a$. We have
\begin{equation*}
\label{mneq}
\begin{split}
	w=\sum_{i=0}^{+\infty}{\pr{Y}{i}[(1-p)u(x, i) + pu(z, i)]}- 
	\sum_{i=0}^{+\infty}{\pr{Y}{i}[(1-q)u(x, i) + q u(y, i)]} \geq 0.
\end{split}
\end{equation*}
We write $w$ as $w = w_1+w_2-w_3-w_4+w_5$, where 
\begin{equation}\nonumber
	\begin{aligned}
		w_1=&\sum_{i=0}^{x}\pr{Y}{i}[[(1-p)u(x, i) + pu(z, i)] - [(1-q)u(x, i) + qu(y, i)]],\\
		w_2=&\sum_{i=x+1}^{+\infty}{\pr{Y}{i}[(1-p)u(x, i) - (1-q)u(x, i)]},\\
		w_3=&\sum_{i=x+1}^{y-{\ut}-1}{\pr{Y}{i}qu(y, i)},\\
		w_4=&\sum_{i=y-{\ut}}^{+\infty}{\pr{Y}{i}qu(y, i)},\\
		w_5=&\sum_{i=x+1}^{+\infty}{\pr{Y}{i}pu(z, i)}.
	\end{aligned}
\end{equation}

%Thus, 
%\begin{equation}\label{maineq}
%  \begin{aligned}
%	&\sum_{i=0}^{x}\pr{Y}{i}[((1-p)u(x, i) + pu(z, i)) - ((1-q)u(x, i) + qu(y, i))] + \\
%	&\sum_{i=x+1}^{+\infty}{\pr{Y}{i}((1-p)u(x, i) - (1-q)u(x, i))} - \sum_{i=x+1}^{y-{\ut}-1}{\pr{Y}{i}qu(y, i)} - \\ 
%	&\sum_{i=y-{\ut}}^{+\infty}{\pr{Y}{i}qu(y, i)} + \sum_{i=x+1}^{+\infty}{\pr{Y}{i}pu(z, i)} \geq 0.
%  \end{aligned}
%\end{equation}
%\label{subeq1}
Since $1-p \geq 1-q$ and $u(z, i)=u(y, i)={\um}$ for all $i \leq x$, we can conclude $w_1 \leq 0$.
%\begin{equation}\label{subeq1}
%  \begin{aligned}
%  w_1 \leq 0
%  \end{aligned}
%\end{equation}
For all $i > x$, we have $u(x, i) \leq 0$, and we also know $1-p \geq 1-q$. These mean $w_2 \leq 0$.
%Therefore
%\begin{equation}\label{subeq1.5}
%w_2 \leq 0
%\end{equation}
Since for all $i \leq y-{\ut}$ we have $u(y, i)={\um}$, we conclude
\begin{equation}\label{subeq2}
	%-\sum_{i=x+1}^{y-{\ut}-1}{\pr{Y}{i}qu(y, i)} \leq 
	w_3 = q{\um}\sum_{i=x+1}^{y-{\ut}-1}{\pr{Y}{i}}
\end{equation}
Moreover, for any arbitrary integers $i$ and $j$, we have $-{\um} \leq u(i, j) \leq {\um}$. Thus
\begin{equation}\label{subeq3}
  \begin{aligned}
-w_4 \leq q{\um}\sum_{i=y-{\ut}}^{+\infty}{\pr{Y}{i}}
  \end{aligned}
\end{equation}
\begin{equation}\label{subeq4}
w_5 \leq p{\um}\sum_{i=x+1}^{+\infty}{\pr{Y}{i}}
\end{equation}
Therefore by knowing $w\geq 0$, $w_1 \leq 0$, $w_2 \leq 0$, and considering Inequalities \ref{subeq2}, \ref{subeq3}, and \ref{subeq4}, we conclude
\begin{equation}\label{eqsun1}
\begin{split}
{\um}(-q\sum_{i=x+1}^{y-{\ut}-1}{\pr{Y}{i}} + q\sum_{i=y-{\ut}}^{+\infty}{\pr{Y}{i}} +  p\sum_{i=x+1}^{+\infty}{\pr{Y}{i}}) \geq w \geq 0.
\end{split}
\end{equation}
$y-{\ut}-1=4b{\ut}+4b-1$ and by Lemma \ref{mins-a}, $\sum_{i=0}^{y-{\ut}-1}{\pr{Y}{i}} \geq 1- \frac{b}{4b{\ut}+4b}=\frac{4{\ut}+3}{4{\ut}+4}$. Note that $x<a$ which means $\sum_{i=0}^{x} \pr{Y}{i} \leq \sum_{i=0}^{a-1} \pr{Y}{i}$. On the other hand, Lemma \ref{maxs-a} says
$\sum_{i=0}^{a-1} \pr{Y}{i} \leq \frac{{\ut}}{{\ut}+1}$. Hence we can conclude $\sum_{i=0}^{x} \pr{Y}{i} \leq \frac{{\ut}}{{\ut}+1}$. Therefore 
\begin{equation}\label{eqsun2}
\sum_{i=x+1}^{y-{\ut}-1}{\pr{Y}{i}} = \sum_{i=0}^{y-{\ut}-1}{\pr{Y}{i}} - \sum_{i=0}^{x}{\pr{Y}{i}} \geq \frac{3}{4{\ut}+4}
\end{equation}
and 
\begin{equation}\label{eqsun3}
\sum_{i=y-{\ut}}^{+\infty}{\pr{Y}{i}} = 1 - \sum_{i=0}^{y-{\ut}-1}{\pr{Y}{i}} \leq \frac{1}{4{\ut}+4}.
\end{equation}

By Inequalities \ref{eqsun1}, \ref{eqsun2}, \ref{eqsun3}, and $\sum_{i=x+1}^{+\infty}{\pr{Y}{i}} \leq 1$, we have %\xxx[Hamid]{I think you have used $\sum_{i=x+1}^{+\infty}{\pr{Y}{i}} \leq \frac{1}{{\ut}+1}$ which is incorrect.}
\begin{equation}\label{eqsun4}
\begin{split}
-q \frac{3}{4{\ut}+4} + q \frac{1}{4{\ut}+4} +  p \geq 0.
\end{split}
\end{equation}
which implies $\frac{q}{p} \leq 2{\ut}+2$. Recalling $p = \alpha^{a}_{z,x}=\frac{a-x}{z-x}$, $q = \alpha^{a}_{y,x}=\frac{a-x}{y-x}$, and $z>y$, we can bound $\frac{z}{y}$ as follows
%\begin{equation*}
$\frac{z}{y} \leq \frac{z-x}{y-x}=\frac{q}{p} \leq 2{\ut}+2$.
%\end{equation*}
Therefore $z \leq y (2{\ut}+2)=\hat{u}$ which is a contradiction. 
Knowing that player $A$ put zero probability on every number $z > \hat{u}$ and considering the definition of \distance\ function $u$, player $B$ will put zero probability of every number greater than $\hat{u}+{\ut}$ in any Nash equilibrium.
%Moreover, by the definition of $u$, increasing the units difference more than ${\ut}$ does not increase the payoff, so
%there is no variable with non-zero probability in $Y$ greater than $8b{\ut}+8b+3{\ut}$.
%\qed
\end{proof}

The following theorem follows immediately after Theorem \ref{finiteTheorem-a} and Lemma \ref{finiteLemma-a}.
\begin{theorem}
\label{thm:dist-a}
%For a given instance of General Lotto game $\Gamma(a, b, u)$ where $u$ is a distance function, 
There is a polynomial time algorithm which finds a Nash Equilibrium of the General Lotto game $\Gamma(a, b, u)$ where $u$ is a \distance\ function.
%a Nash equilibrium can be found in time $\mathcal{O}(
%{\poly}(a, b))$.
\end{theorem}
\begin{proof}
Let $\bar{u}=(4b{\ut}+4b+{\ut})(2{\ut}+2)+{\ut}$. Lemma \ref{finiteLemma-a} shows there is a bound on the optimal strategies in a Nash equilibrium. More precisely $Pr(Y>\bar{u})=0$, where $Y$ is a strategy of player $A$ or $B$. Thus General Lotto game $\Gamma(a, b, u)$ is equivalent to finite General Lotto game $\Gamma(a, b, f, S)$, where $S=\{1, 2, \ldots, \bar{u}\}$. By Theorem \ref{finiteTheorem-a}, a polynomial-time algorithm finds a Nash equilibrium of the game.
\end{proof}

\section{Oracles}\label{oracles}

In this section we describe, in precise detail, the separating oracles used by the ellipsoid method to solve our represented linear programs.
Consider we are given a sequence $c_0, c_1,\ldots, c_{k(m+1)}$, where $k$ is the number of battlefields and $m$ is the number of troops for a player. We first present an algorithm which finds a pure strategy $x=(x_1, x_2, \dots, x_k)\in \pureset$ such that $\sum_{i=1}^kx_i=m$, and $\xh=\mathcal{G}(x)$ minimizes the following equation.
\begin{equation}\label{minpure_eq-a}
%c_0+c_1 \hat{x}_1 + \ldots+ c_n \hat{x}_n
c_0+\sum_{i=1}^{k(m+1)} c_i \hat{x}_i
\end{equation}
Then we leverage this algorithm and design polynomial-time algorithms for the hyperplane separating oracle and best-response separating oracle. 
The following lemma shows that Algorithm \ref{alg:minpure} (\minpure) finds the minimizer of Equation \ref{minpure_eq-a}.
\begin{algorithm}
\caption{\minpure} 
\textbf{input:} $m$, $k$, $c_0,c_1,c_2,\ldots,c_{k(m+1)}$
%\textbf{output:} $salam$
\begin{algorithmic}[1]

\For{$j\gets 1$ to $m$}
	\State $d[0,j] \gets c_0$
\EndFor
\For{$i\gets 1$ to $k$}
	\For{$t\gets 0$ to $m$}
		\For{$j\gets 0$ to $t$}
			\If{$d[i-1,t-j]+c_{(i-1)(m+1)+j+1} <  d[i,t]$}
			\State $d[i,t]\gets d[i-1,t-j]+c_{(i-1)(m+1)+j+1}$
			\State $r[i,t]\gets j$
			\EndIf
		\EndFor
	\EndFor
\EndFor

%\If {$d_{k, m} \geq 0$}
%	\State \Return accept
%\Else
\State $rem \gets m$
\For{$i \gets k$ downto $1$}
	\State $x_i\gets r[i,rem]$
	\State $rem \gets rem-r[i,rem]$
\EndFor
\State \Return $x=(x_1,x_2,\ldots,x_k)$
%\EndIf
\end{algorithmic}
\label{alg:minpure}
\end{algorithm}

\begin{lemma}\label{lem:minpure-a}
Given two integers $m$ and $k$ and a sequence $c_0,c_1,\ldots,c_{k(m+1)}$, algorithm \minpure\ correctly finds an optimal pure strategy $x=(x_1, x_2, \dots, x_k)$ where $\sum_{i=1}^{k}x_i=m$, $\xh=\mathcal{G}(x)$ and $\xh$ minimizes $c_0+\sum_{i=1}^{k(m+1)} c_i \hat{x}_i$. %such that Equation \ref{minpure_eq} is minimized.
\end{lemma}
\begin{proof}
%Algorithm \minpure uses a dynamic programming approach to solve that problem.
In Algorithm \minpure, using a dynamic programming approach, we define $d[i,t]$ to be the minimum possible value of $c_0+\sum_{i'=1}^{i(t+1)} c_{i'}\hat{x}_{i'}$ where $\sum_{i'=1}^{i}x_{i'}=t$. Hence, $d[k,m]$ denotes the minimum possible value of $c_0+\sum_{i=1}^{k(m+1)} c_i \hat{x}_i$. Now, we show that Algorithm \minpure\ correctly computes $d[i,t]$ for all $0\leq i \leq k$ and $0\leq t\leq m$. Obviously $d[0,j]$ is equal to $c_0$. For an arbitrary $i>0$ and $t$, the optimal strategy $x$ puts $0\leq t'\leq t$ units in the $i$-th battlefield and the applied cost in the equation \ref{minpure_eq-a} is equal to $c_{(i-1)(m+1)+t'+1}$. Thus,
\begin{equation*}
d[i,t] = \min_{0\leq t'\leq t}\{d[i-1,t-t']+c_{(i-1)(m+1)+t'+1}\}
\end{equation*}
To compute the optimal pure strategy $x=(x_1,x_2,\ldots,x_k)$ we also keep a value 
\begin{equation*}
r[i,t] = \argmin_{0\leq t'\leq t}\{d[i-1,t-t']+c_{(i-1)(m+1)+t'+1}\}
\end{equation*}
which determines the number of units the optimal strategy should put in the $i$-th battlefield to minimize $c_0+\sum_{i'=1}^{i(t+1)} c_{i'}\hat{x}_{i'}$. Assuming we have correctly computed $x_{i+1},\ldots,x_k$, in line $16$, algorithm \minpure\ correctly computes $x_i$ which is equal to $r[i,m-\sum_{j=i+1}^{k} x_j]$. Since $x_k=r[k,m]$ we can conclude algorithm \minpure\ correctly computes the optimal strategy $x=(x_1, x_2, \dots, x_k)$.
%We define subproblem $\minpure(i,t)$ to be the problem of finding $\hx_1,\hx_2,\ldots,\hx_{j(a+1)}$ such that  $c_0+\sum_{i=1}^{j(a+1)} c_i\hat{x}_i$ is minimized.
\end{proof}

\subsection{Hyperplane separating oracle}
\label{oracle:sep}
Algorithm \ref{alg:hp_sep} (\hporacle) gets a hyperplane as input and either finds a point in $\vs$ which violates constraints in LP \ref{Prog:separating_plane} or reports that all points in $\vs$ are satisfying all constraints in LP \ref{Prog:separating_plane}. We suppose that the input hyperplane is described by the following equation,
\begin{equation}
\alpha_0+\alpha_1 \hat{x}_1 + \ldots+ \alpha_{k(a+1)} \hat{x}_{k(a+1)}=0
\end{equation}

and we want to find a point $\xh\in \vs$ that violates the following constraint:
\begin{equation}
\alpha_0+\sum_{i=1}^{n} \alpha_i \hat{x}_i \geq 0
\end{equation}

This problem is equivalent to finding a point $\xh^{min} \in \vs$ which minimizes equation $\alpha_0+\sum_{i=1}^{n} \alpha_i \hat{x}_i$. If $\alpha_0+\sum_{i=1}^{n} \alpha_i \hat{x}^{min}_i \geq 0$ it means all points in $\vs$ are satisfying the constraints of LP, and otherwise $\xh^{min}$ is a point which violates constraint \ref{cn2}. Since points in $\vs$ are equivalent to pure strategies of player $A$, we can use algorithm \minpure\  to find $\xh^{min}$. Thus we can conclude algorithm \hporacle\ correctly finds a violated constraint or reports that the hyperplane satisfies constraints \ref{cn2} of LP \ref{Prog:separating_plane}.
%actually we want to find a pure strategy which minimizes $\alpha_0+\sum_{i=1}^{n} \alpha_i \hat{x}_i$. To find such strategy, we use a dynamic programming approach. 

\begin{algorithm}
%\caption{Hyperplane Separating Oracle} 
\caption{\hporacle} 
\textbf{input:} $a$, $k$, $\alpha_0,\alpha_1,\ldots,\alpha_{k(a+1)}$
\begin{algorithmic}[1]
\State $x^{min} \gets \minpure(a,k,\alpha_0,\alpha_1,\ldots,\alpha_{k(a+1)})$
\State $\xh^{min} \gets \G(x^{min})$
\If{$\alpha_0+\sum_{i=1}^{k(a+1)} \alpha_i \hat{x}^{min}_i \geq 0$}
\State \Return pass
\Else
\State \Return $\alpha_0+\sum_{i=1}^{k(a+1)} \alpha_i \hat{x}^{min}_i < 0$
\EndIf
\end{algorithmic}
\label{alg:hp_sep}
\end{algorithm}

\subsection{Best-response separating oracle}
\label{oracle:best}
Algorithm \ref{alg:br_sep} (\broracle) gets a pair $(\xh,U)$ as input and decides whether there is a pure strategy $y=(y_1,y_2,\ldots,y_k)\in \mathcal{Y}$ such that for $\yh=\Gb(y)$, we have
\begin{equation}\label{brequation}
\sum_{i=1}^{k} \sum_{t_a=0}^{a} \sum_{t_b=0}^{b} \hat{x}_{i,t_a} \hat{y}_{i,t_b} u^A_i(t_a, t_b) < U.
\end{equation}

We can rewrite inequality \ref{brequation} as follows:
\begin{equation*}
\sum_{i=1}^{k} \sum_{t_b=0}^{b} \hat{y}_{i,t_b} \sum_{t_a=0}^{a} \hat{x}_{i,t_a} u^A_i(t_a, t_b) < U
\end{equation*}

Therefore, by letting $c_{i,t_b} = \sum_{t_a=0}^{a} \hat{x}_{i,t_a} u^A_i(t_a, t_b)$, this problem is equivalent to find a point $\yh^{min}\in \Ib$ which minimizes $\sum_{i'=1}^{k(b+1)} c_{i'} \hat{y}_{i'}$. If $\sum_{i'=1}^{k(b+1)} c_{i'} \hat{y}^{min}_{i'} < U$ we have found a violating payoff constraint of LP \ref{Prog:3} and if $\sum_{i'=1}^{k(b+1)} c_{i'} \hat{y}^{min}_{i'} \geq U$, pair $(\xh,U)$ satisfies all the payoff constraints of LP \ref{Prog:3}. Thus, by Lemma \ref{lem:minpure-a} we conclude that algorithm \broracle\ correctly finds a violating payoff constraint of LP \ref{Prog:3} or reports that $(\xh,U)$ satisfies all the payoff constrains.

\begin{algorithm}
%\caption{Best-Response Separating Oracle} 
\caption{\broracle} 
\textbf{input:} $a$, $b$, $k$, $U$, $\hx_1,\ldots,\hx_{k(a+1)}$

\begin{algorithmic}[1]
\For{$i\gets 1$ to $k$}
\For{$t_b\gets 0$ to $b$}
%\State $c_{i,t_b} \gets 0$
%\For{$t_a\gets 0$ to $a$}
\State $c_{(i-1)(b+1)+t_b+1}=c_{i,t_b} = \sum_{t_a=0}^{a} \hx_{i,t_a}$ $u_{i}^{A}(t_a,t_b)$
%\EndFor
\EndFor
\EndFor
\State $y^{min} \gets \minpure(b,k,c_0,c_1,c_2,\ldots,c_{k(b+1)})$
\State $\yh^{min} \gets \Gb(y^{min})$
\If{$\sum_{i=1}^{k(b+1)} c_i \hat{y}^{min}_i \geq U$}
\State \Return pass
\Else
\State \Return $\sum_{i=1}^{k(b+1)} c_i \hat{y}^{min}_i < U$
\EndIf
\end{algorithmic}
\label{alg:br_sep}
\end{algorithm}

\section{Proofs Ommitted from Appendix A}
\label{app:omitted}

\begin{proofof}{Lemma \ref{polyhedron}}
We need Lemmas \ref{linear-property} and \ref{convexity-lemma} for proving Lemma \ref{polyhedron}.
\begin{lemma}
\label{linear-property}
Let $\x^1,\x^2,\ldots,\x^t$ be  $t$ arbitrary mixed strategies of player $A$,  $\sum_{r=1}^t \alpha_r =1$, and  $\x = \sum_{r=1}^t \alpha_r\x^r$ be a mixed strategy that plays strategy $\x^r$ with probability $\alpha_r$,
then $\functiona(\x) = \sum_{r =1}^t \alpha_r\functiona(\x^r)$.
\end{lemma}
\begin{proof}
%Let $\xh = \functiona(\x)$. 
Since $\x = \sum_{r=1}^t \alpha_r\x^r$ and $\functiona(\x^r)_j$ represent the probability that strategy $\x^r$ plays $j$, we have $[\functiona(\x)]_j = \sum_{r=1}^t\alpha_r[\functiona(\x^r)]_j$,  for all  $1 \leq j \leq \na$. Therefore, $\functiona(\x) = \sum_{r=1}^t\alpha_r\functiona(\x^r)$.
%\qed
\end{proof}

%Now we prove some properties of $\sets$ in order to us
\begin{lemma}\label{convexity-lemma}
$\sets$ is a convex set.
\end{lemma}
\begin{proof}
A set of points is convex if and only if every segment joining two of its points  is completely in the set. Let $\ppointa{\xh}{\hat{x}}$ and $\ppointa{\xhp}{\hat{x}'}$ be two points in $\Rna$. We show that if $\xh,\xhp \in \sets$ then for every $0 \leq \alpha \leq 1$, $\xhz =  (\alpha \xh + (1-\alpha) \xhp) \in \sets$.\\
Since $\xh$ and $\xhp$ are in $\sets$, there exist mixed strategies $\x$ and $\xp$ in $\strategyset$, such that $\xh = \functiona(\x)$ and $\xhp = \functiona(\xp)$. Let $\xz = \alpha \x + (1-\alpha) \xp$ be a mixed strategy of player $A$ that plays $\x$ with probability $\alpha$ and $\xp$ with probability $1-\alpha$. By Lemma \ref{linear-property} we have $\functiona(\xz) = \alpha \xh + (1-\alpha) \xhp = \xhz$ hence, $\xhz \in \sets$. 
%\qed
\end{proof}

Now we are ready to proof Lemma \ref{polyhedron}. By Lemma \ref{convexity-lemma}, we know  $\sets$ is convex. We show that we can find a finite set of points in $\sets$, such that every point in $\sets$ can be written as a convex combination of these points. Note that, $\pureset = \{\pureset^1,\pureset^2,\ldots,\pureset^{|\pureset|}\}$ is the set of all pure strategies of player $A$ and $\vs = \{\vv^1,\vv^2,\ldots,\vv^{|\pureset|}\}$ is the set of points where $ \vv^i =\functiona(\pureset^i) $. Note that, $\vs$ and $\pureset$ are finite sets. Every strategy of player $A$ is either a pure strategy or a mixed strategy and can be written as a convex combination of the pure strategies. Therefore, according to Lemma \ref{linear-property} every point in $\sets$ is either in $\vs$ or can be written as a convex combination of points in $\vs$. Hence, $\sets$ forms a convex polyhedron in $\Rna$.
%\qed
%\end{proof}

Next we show that the number of vertices and facets of $\sets$ is $\cpx{2^{\poly(n)}}$.
%\begin{lemma}\label{vertex-number}
%Polyhedron $\sets$ has no more than ${\na}^{\na}$ vertices.
%\end{lemma}
%\begin{proof}
Let $\xh$ be a vertex of polyhedron $\sets$ and $\x$ be a strategy of player $A$ such that $\functiona(\x) = \xh$. Since $\xh$ is a vertex of $\sets$, it cannot be written as a convex combination of other vertices of $\sets$. If $\x$ is a mixed strategy, it can be written as a convex combination of other strategies and according to Lemma \ref{linear-property}, $\functiona(\x)$ can be written as  convex combination of other points of $\sets$. This implies that either $\x$ is a pure strategy or there exists a pure strategy $\xp$ such that $\functiona(\xp) = \xh$. Therefore, the number of vertices of $\sets$ is no more than the number of pure strategies of  player $A$. Since each pure strategy of player $A$ is a partition of $a$ units into $k$ battlefields the number of pure strategies is no more than $(a+1)^k$. Thus, the number of vertices of $\sets$ is at most ${\na}^{\na}$.
%\qed
%\end{proof}

%\begin{lemma}\label{facet-number}
%Polyhedron $\sets$ has no more than $\na^{(\na^2)}$ facets.
%\end{lemma}
%\begin{proof}
Let $d \leq \na$ be the dimension of $\sets$. Since every facet of $\sets$ can be uniquely determined by $d$ vertices of $\sets$, the number of facets is no more than $$\binom{|\vs|}{d} \leq (|\vs|)^d \leq ({\na}^{\na})^d \leq \na^{(\na^2)} .$$ %Therefore, this value is at most $\na^{(\na^2)}$.%\qed
\end{proofof}

\begin{proofof}{Lemma \ref{lem:poly_constraint}}
Note that the dimension of $\sets$ is not necessarily $\na$. Let $d$ be the dimension of $\sets$ and $\spt$ be the affine hull of $\sets$. Since $\spt$ is a $d$-dimensional subspace of $\Rna$, it can be represented as the intersection of $\na-d$ orthogonal hyperplanes. Let $L$ be a set of such hyperplanes.
Let $C$ be the set of all hyperplanes perpendicular to hyperplanes in $L$ which contain a facet of $\sets$.
We need Lemma \ref{li0}, Lemma \ref{li1}, Lemma \ref{li4}, and Lemma \ref{totlem} for proving Lemma \ref{lem:poly_constraint}.

\begin{lemma}\label{li0}
There exists exactly one hyperplane perpendicular to all hyperplanes in $L$ which contains all points of a $(d-1)$-dimensional subspace of $\spt$.
\end{lemma}
\begin{proof}
Let $O$ be the orthogonal basis of the $(d-1)$-dimensional subspace and $N$ be the set of normal vectors of hyperplanes in $L$. Since every vector in $O$ is in $\spt$, all vectors in $O$ are orthogonal to all vectors in $N$. Therefore, the desired hyperplane should be orthogonal to all vectors $N\cup O$. Since $|N \cup O| = \na-1$ and all vectors in $N \cup O$ are pairwise orthogonal, there exists exactly one hyperplane containing the subspace and perpendicular to all hyperplanes in $L$.
\end{proof}

\begin{lemma}\label{li1}
There exists exactly one hyperplane perpendicular to all hyperplanes of $L$ which contain all points of some facet $f$ of $\sets$.
\end{lemma}
\begin{proof}
Since $f$ is a facet of $\sets$, its dimension is $d-1$. %\xxx[Hamid]{Ref}
 Therefore, its affine hull is a $(d-1)$-dimensional subspace of $\spt$. Note that, a hyperplane contains $f$ if and only if it contains affine hull of $f$. Lemma \ref{li0} states that there exists exactly one hyperplane perpendicular to all hyperplanes of $L$ that contains the affine hull of $f$. Hence, there exists a unique hyperplane containing $f$ which is perpendicular to all hyperplanes of $L$.
%\qed
\end{proof}
%\begin{lemma}\label{li3}
%Let $P \in S'$ be a point which is not in $S$. There exists a hyperplane that is perpendicular to all hyperplanes of $L$ and separates $P$ from $S$.
%\end{lemma}
%\begin{proof}
%Since $P$ is not in $S$, there exists a hyperplane $H$ that separates $P$ from $S$. Let $H'$ be the projection of $H$ on $S'$. $P$ is in not in $H'$, therefore $H'$ does not contain all points of $S'$. Thus, $H'$ is a $(d-1)$-dimensional subspace of $S'$ that separates $P$ from $S$. By Lemma \ref{li0}, we know that there exists a unique hyperplane that contains $H'$ and is perpendicular to all hyperplanes of $L$. Since it contains $H'$, it separates $P$ from $S$.
%\qed
%\end{proof}

%\begin{lemma}\label{li2}
%Let $O$ be an interior point of $S$ and $p \in S'$ be a point that is not in $S$. Let $Q$ be a hyperplane perpendicular to all hyperplanes of $L$ that separates $p$ from $S$ and its distance to $O$ is minimal. There exists a facet of $S$ that is compeletely in $Q$.
%\end{lemma}
%\begin{proof}
%Consider $q \in Q$ be a point in $Q$ which is nearest to $O$. Note that, $Q$ is perpendicular to all hyperplanes of $L$, thus $q$ is in $S'$. Since all points of $S$ are in one side of $Q$, $q$ cannot be an interior point of $S$. Furthermore, if $q$ is not in $S$, we can shift $Q$ towards $S$ and obtain a separating hyperplane perpendicular to hyperplanes of $L$ which has less distance from $O$. Hence, $q$ is in at least one facet of $S$. Therefore there exists one facet of $S$ that contains $q$ and separates $p$ from $S$, namely $f$.
%\qed
%\end{proof}

\begin{lemma}\label{li4}
For every point $\xh \in \spt$ which is not in $\sets$, there exists a hyperplane perpendicular to all hyperplanes in $L$, which contains a facet of $\sets$ and separates $\xh$ from $\sets$.
\end{lemma}
\begin{proof}
Consider a segment between $\xh$ and one point of $\sets$. Since one endpoint of the segment is in $\sets$ and the other one is not, it has intersection with at least one facet of $\sets$, namely $f$. %\xxx[Hamid]{Ref}
 By Lemma \ref{li1}, there exists one hyperplane which contains $f$ and is perpendicular to hyperplanes in $L$. This hyperplane has intersection with the segment and separates $\xh$ from $\sets$.
\end{proof}

\begin{lemma}
\label{totlem}
Point $\xh \in \Rna$ is in $\sets$ if and only if all hyperplanes in $L$ contain $\xh$ and no hyperplane in $C$ separates $\xh$ from $\sets$.
\end{lemma}
\begin{proof}
Since $\sets$ is in the intersection of all hyperplanes in $L$, if $\xh \in \sets$ it is also in all hyperplanes in $L$ and obviously no hyperplane in $C$ separates $\xh$ from $\sets$. Now suppose $\xh \notin \sets$, if $\xh \notin \spt$, then $\xh$ is not in all hyperplanes of $L$, otherwise by Lemma \ref{li4} there exists a hyperplane in $C$ that separates $\xh$ from $\sets$. 
\end{proof}

Now we are ready to prove Lemma \ref{lem:poly_constraint}. By Lemma \ref{totlem} we conclude that set of hyperplanes in $L$ and $C$ are sufficient to formulate $\sets$.
We know $|L|$ is at most $\na$, and by Lemma \ref{li1} we can find out that $|C|$ is equal to the number of facets of $\sets$. Moreover Lemma \ref{polyhedron} states that $|C|$ is $\cpx{2^{\poly(n)}}$ which means $|C| + |L| \in \cpx{2^{\poly(n)}}$. 
%\xxx[Hamid]{Check it.} 
\end{proofof}

\section{Approximating the payoff of the game}
\label{sec:approx}

We can observe that in some dueling games the separation problem cannot be solved in polynomial time. However, having a fixed error threshold, it is possible to approximately solve the separation problem in polynomial time. In this section we present a technique for approximating the payoff of the dueling games via providing an approximate solution for the separation problem.

We say a (not necessarily feasible) point $\xh$ is an $\epsilon$-solution of the LP in a bilinear dueling game where there exists a feasible strategy point $\xh'$ such that $|\xh-\xh'| \leq \epsilon$ and the minimum payoff of strategy $\xh$ differs by at most $\epsilon$ from that of the MinMax strategies. 
Remark that, in the separation problem we are given a linear constraint $\alpha_0,\alpha_1,\ldots$ and are to find out whether there exists a strategy point for a player which has the property $\alpha_0+\sum_i\alpha_i \hat{x}_i \geq 0$. We say a separation oracle solves this problem with approximation factor $\epsilon$ if it always finds such a point when there exists a violating point having a distance of at least $\epsilon$ from the hyperplane. Note that, the approximate oracle may not report a correct answer when the distance of all desired points to the hyperplane is less than $\epsilon$. We say a dueling game is \textit{polynomially $\epsilon$-separable} if we can approximately solve the separation problem in polynomial time for every $\epsilon > 0$. The following theorem provides a strong tool for approximating the payoff of the game in an NE for a broad set of dueling games.

\begin{theorem}\label{apptheorem}
Let $B$ be a bilinear dueling game ranging over $[0,1]^{\na}\times[0,1]^{\nb}$ (i.e. $\Sa \subset [0,1]^{\na}$ and $\Sb \subset [0,1]^{\nb}$) with payoff matrix $M$ and $n = \max\{\na,\nb\}$. If we can solve the $\epsilon$-separation problem in time $f(\epsilon,n)$ then we can find an $\epsilon$-solution of the LP in time $\cpx{\poly(n)\cdot f(\epsilon/\beta,n)}$, where $\beta = \max\{1,\sum_{i,j}|M_{i,j}|\}$.
\end{theorem}
\begin{proof}
Suppose we solve LP \ref{Prog:3} with Ellipsoid method using an $\epsilon/\beta$-separating oracle instead of the exact oracle and let $\xh$ be the solution we find. First we show that there exists at least one feasible strategy $\xh'$ such that $|\xh'-\xh| \leq \epsilon$. Next, we will compare the minimum payoff of $\xh'$ with that of the MinMax strategies.
Remark that, in order to determine whether or not $\xh$ is within the set of feasible strategy points we solve the following linear program (which is the same as LP \ref{Prog:separating_plane} from the proof of Lemma \ref{lem:membership}) and see if there is any violating constraint.
\begin{eqnarray}
%\label{Prog:separating_plane_approx}
%\min & 0&\\
%s.t. &
\alpha_0 + \sum_j \alpha_{j} {\hat{x}}_{j} < 0& \\
\alpha_0 + \sum_j \alpha_{j} {\hat{v}}_{j} \geq 0 & \forall \mathbf{\hat{v}} \in \Ia \nonumber% \\
%&\sum_j\alpha_{j} = 1,& \nonumber
\end{eqnarray}
Since the $\epsilon/\beta$-separating oracle does not report any violating constrains for $\xh$, either there is no violating constrain or the distance of $\xh$ from all violating hyperplanes is less than $\epsilon/\beta$. Therefore there exists a point $\xh'$ in the set of feasible solution such that $|\xh'-\xh| \leq \epsilon/\beta \leq \epsilon$.\\
Next, we compare the minimum payoff of strategy $\xh'$ with the minimum payoff of the MinMax strategies. From the previous argument we have $|\xh'-\xh| < \epsilon/\beta$. Since $\beta = \sum_{i,j} |M_{i,j}|$, we have $|(\xh M \yh) - (\xh' M \yh)| < |\xh'-\xh|.\beta < \epsilon$ for each strategy $\yh$ of the second player. Moreover, since every MinMax strategy is a potential solution of the linear program and $\xh$ is the solution which maximizes the objective function, the minimum possible payoff off $\xh$ is not less than the minimum possible payoff of the MinMax strategies. Thus, the difference between the minimum possible payoff of $\xh'$ and that of the MinMax strategies is at most $\epsilon$.\\
Since we're using Ellipsoid method, the number of times we call the $\epsilon$-separation oracle is $\poly(n)$, hence the running time of the algorithm is $\cpx{\poly(n)\cdot f(\epsilon/\beta,n)}$.
\end{proof}

Note that, Theorem \ref{apptheorem} states if we can find an approximate solution of the separation problem in polynomial time, then we can find an approximate solution of the LP in polynomial time as well. Since an approximate solution  may not necessarily be a feasible strategy point, it cannot be used in order to characterize the properties of the MinMax strategies. However, we can approximate the payoff of the players in an NE  by computing the payoffs for the $\epsilon$-solution of the LP. 

\end{document}